\documentclass[sigconf]{acmart}

\usepackage{graphicx}
\usepackage{amssymb}
\usepackage{amsthm}
\usepackage{amsmath}
\usepackage{xfrac}
\usepackage{subcaption}
\usepackage{pgfplots}
\pgfplotsset{compat=1.13}
\usepgfplotslibrary{fillbetween}
\usetikzlibrary{positioning,patterns,shapes}
\usepackage{booktabs}
\usepackage{comment}
\usepackage{macros}
\usepackage[ruled,linesnumbered]{algorithm2e}
\usepackage{multicol,multirow,dcolumn,longtable,footnote}
\usepackage{siunitx}
\usepackage{xspace}
\usepackage[flushleft,para]{threeparttable}

\theoremstyle{definition}

\theoremstyle{plain}

\newtheorem{proposition}{Proposition}
\theoremstyle{remark}

\theoremstyle{plain}

\newif\ifextendedversion
\extendedversiontrue

\setcopyright{none}

\newcolumntype{P}[1]{D{.}{.}{#1}} 
\newcolumntype{V}[1]{>{$\vcenter\bgroup\hbox\bgroup}#1<{\egroup\egroup$}} 

\begin{document}
\title[Reach Set Approximation through Decomposition]{Reach Set Approximation through Decomposition with Low-dimensional Sets and High-dimensional Matrices}

\author{Sergiy Bogomolov}
\affiliation{%
	\institution{Australian National University}
	\city{Canberra}
	\state{Australia}
}

\author{Marcelo Forets}
\author{Goran Frehse}
\author{Fr\'ed\'eric Viry}
\affiliation{%
	\institution{Univ. Grenoble Alpes, VERIMAG}
	\city{Grenoble}
	\state{France}
}

\author{Andreas Podelski}
\author{Christian Schilling}
\affiliation{%
	\institution{University of Freiburg}
	\city{Freiburg}
	\state{Germany}
}

\renewcommand{\shortauthors}{S. Bogomolov et al.}

\begin{abstract}
%
Approximating the set of reachable states of a dynamical system is an algorithmic yet  mathematically rigorous way to reason about its safety.
Although progress has been made in the development of efficient algorithms for affine
dynamical systems, available algorithms still lack scalability to ensure their wide adoption in the industrial setting.
While modern linear algebra packages are efficient for matrices with tens of thousands of dimensions, set-based image computations are limited to a few hundred.
We propose to decompose reach set computations such that set operations are performed in low dimensions, while matrix operations like exponentiation are carried out in the full dimension.
Our method is applicable both in dense- and discrete-time settings. For a set of standard benchmarks, it shows a speed-up of up to two orders of magnitude compared to the respective state-of-the art tools, with only modest losses in accuracy.
For the dense-time case, we show an experiment with more than 10.000~variables, roughly two orders of magnitude higher than possible with previous approaches.

\end{abstract}


\maketitle

\section{Introduction}
\label{sec:introduction}
%
Verifying safety properties for dynamical systems is an important and intricate task.
For bounded time it is well known that the problem can be reduced to the computation of the reachable states.
We are interested in the set-based reachability problem for affine dynamical systems~\cite{GirardLGM06}.
Here, recurrence relations of the form
\begin{equation}\label{eq:discrete_affine_recurrence}
\X(k+1)= \Phi \X(k) \oplus \V(k),\quad k = 0,1,\ldots,N
\end{equation}
arise naturally. In the context of control engineering, the sequence of sets $\{\V(k)\}_k$ usually represents contributions from nondeterministic inputs or noise, $\oplus$ denotes the Minkowski sum between $n$-dimensional sets, $\Phi$ is a given real \nbyn matrix, and the set $\X(0)$ accounts for uncertain initial states.

Numerous works present strategies for solving
equation~\eqref{eq:discrete_affine_recurrence} in the form of
ellipsoids~\cite{KurzhanskiV00, kurzhanskiy2006ellipsoidal}, template
polyhedra such as zonotopes~\cite{Girard05, althoff2014reachability}
or support functions~\cite{LeGuernic2010250, frehse2011spaceex,
FrehseKLG13,DBLP:conf/hybrid/FrehseBGSP15,bogomolov-et-al:tacas-2017},
or a combination~\cite{althoff2016combining}. The problem also
generalizes to hybrid systems with piecewise affine
dynamics~\cite{AsarinDMB00, LeGuernicG09}. A key difficulty is
scalability, as the cost of some set operations increases
superlinearly with the dimension. A second challenge is the error
accumulation for increasing values of~$N$, known as the wrapping
effect.

\smallskip

In this paper, we propose, analyze, and evaluate a novel \emph{partial decomposition algorithm} for solving equation~\eqref{eq:discrete_affine_recurrence} that does not suffer from the wrapping effect if the inputs are held constant over all time.
The complexity  of non-decomposition approaches is mostly affected by the dimension~$n$ and grows superlinearly with it.
Our method partially shifts this dependence on~$n$ to other structural properties:
we perform set operations in low dimensions (unaffected by~$n$);
we effectively omit variables from the analysis if they are not involved in the property of interest;
and we exploit the sparsity of~$\Phi$ and its higher-order powers.
However, unlike other decomposition approaches, we keep the matrix computations in high dimensions, which allows us to produce precise approximations.
The strategy consists of decomposing the discrete recurrence relation~\eqref{eq:discrete_affine_recurrence} into subsystems of low dimensions.
Then we compute the reachable states for each subsystem; these low-dimensional set operations can be performed efficiently.
Finally we compose the low-dimensional sets symbolically and project onto the desired output variables.
The analysis scales to systems with tens of thousands of variables, which are out of scope of state-of-the-art tools for dense-time reachability. 

\smallskip

We apply our method to compute reachable states and verify safety properties of affine dynamical systems,
\begin{equation}\label{eq:continuous_system}
 x'(t)= A x(t) + Bu(t).
\end{equation}
The initial state can be any point in a given set $\X_0$, and $u(t) \in \U(t) \subset \Reals^m$ is a nondeterministic input. Both the initial set and the set of input functions are assumed to be compact and convex. We also consider observable outputs,
\begin{equation}\label{eq:output_system}
y(t)= C x(t) + Du(t),
\end{equation}
where $C$ and $D$ are matrices of appropriate dimension. In mathematical systems theory,  equations~\eqref{eq:continuous_system}-\eqref{eq:output_system} define what is known as a linear time-invariant (LTI) system.

\paragraph{Contribution}

We present a new method to solve the reachability problem for affine dynamical systems with nondeterministic inputs and experimentally show that it is highly scalable under modest loss of accuracy.
More precisely:

\begin{itemize}

\item We provide a new decomposition approach to solve equation~\eqref{eq:discrete_affine_recurrence} and analyze the approximation error.

\item We address both the dense time and the discrete time instances of the reachability problem for general LTI systems of the form~\eqref{eq:continuous_system}-\eqref{eq:output_system}.

\item We implement our approach efficiently and demonstrate its scalability on real engineering benchmarks.
The tool, source code, and benchmark scripts are publicly available~\cite{tool}.
\end{itemize}

\paragraph{Related work}

Kaynama and Oishi consider a Schur-based decomposition to compute the reachable states~\cite{kaynama2009schur, kaynama2010overapproximating, kaynama2011complexity}. They approximate the result for subsystems by nondeterministic inputs using a static (i.e., time-unaware) box approximation. The authors also address approximation errors by solving a Sylvester equation to obtain a similarity transformation that minimizes the submatrix coupling.

For systems where variables are linearly correlated in the initial states and inputs are constant, Han and Krogh propose an approximation method that uses Krylov subspace approximations~\cite{HanK06} without explicitly decomposing the system.

If the system is singularly perturbed with different time scales (``slow and fast variables''), time-scale decomposition can be applied~\cite{Dontchev92,GoncharovaO09}. We do not consider this setting here.

The reachability analysis tool Coho uses \emph{projectahedra} -- an approximate polyhedron representation consisting of all possible axis-aligned 2D projections -- for set representation~\cite{GreenstreetM99,YanG08a}.

Seladji and Bouissou define a sub-polyhedra abstract domain based on support functions~\cite{SeladjiB13}.
Our approach can choose directions dynamically, and independently for each subsystem.

An orthogonal approach to reduce the complexity of system analysis is known as \emph{model order reduction} (MOR)~\cite{antoulas2001survey}.
The idea is to construct a lower-dimensional model with \emph{similar} behavior.
Recently there have been efforts to combine MOR and abstraction
techniques to obtain a sound overapproximation~\cite{TranNXJ17}. In a
further approach, Bogomolov et al.~\cite{DBLP:conf/atva/BogomolovMP10}
suggest an abstraction technique, which employs dwell time
bounds. Moreover,
Bogomolov~et~al.~\cite{DBLP:conf/hybrid/BogomolovHMWP14} introduce a
system transformation to reduce the state space dimension based on the
notion of quasi-dependent variables, which captures the dependencies
between system state variables. In principle, such methods could be
used as a preprocessing for our approach, where the approximation
errors would then be combined.

Bak and Duggirala check safety properties and compute counterexample traces for LTI systems in a ``simulation equivalent manner''~\cite{bak2017simulation}. A reachable set computed in this way consists of all the states that can be reached by a fixed-step simulation for any choice of the initial state and piecewise constant input. This set, however, does not include all trajectories of equation~\eqref{eq:continuous_system}. The simulation equivalent reachability also involves a recurrence of the type~\eqref{eq:discrete_affine_recurrence}, and we study its decomposed form in this work as well. 

Decomposition methods have also been designed for the reachability problem of nonlinear ODEs. Chen et al.\ show that, using Hamilton-Jacobi methods, the (analytically) exact reachable states can be reconstructed from an analysis of the subsystems for general ODE systems~\cite{ChenHT17}.
The system needs, however, be composed of so-called \emph{self-contained subsystems}, which is a strong assumption.
The technique is based on~\cite{MitchellT03} which has no such limitation but suffers from a projection error.
For general LTI systems (which we consider) an approximation error is unavoidable.

Asarin and Dang propose a decomposition approach where they project away variables and abstract them by time-unaware differential inclusions~\cite{AsarinD04}.
To address the overapproximation, they split these variables again into several subdomains.

Chen and Sankaranarayanan apply uniform hybridization to analyze the subsystems over time and feed the results to the other subsystems as time-varying interval-shaped inputs~\cite{ChenS16decomposed}.
In contrast, our reachability algorithm needs not be performed iteratively because the analysis of each subystem is completely decoupled.

Schupp et al.\ decompose a system by syntactic independence~\cite{SchuppNA17}.
In our setting this corresponds to models where the dynamics matrix has a block diagonal form.
For such systems the dynamical error is zero in both their and our approach.

\smallskip

The paper is organized as follows. In Section~\ref{sec:preliminaries} we recall some basics on approximating convex sets with polyhedra, which motivate our choice for decomposing into two dimensions. We also recall a state-of-the art algorithm for approximating the reach sets of affine systems using the affine recurrence relation \eqref{eq:discrete_affine_recurrence}. In Section~\ref{sec:decomposition}, we start by considering the decomposition of a single affine map, and then develop the more general case of an affine recurrence.
The approximation error is discussed in Section~\ref{ssec:approximation_error}.
We present our reachability algorithm in Section~\ref{sec:implementation}, discuss the different techniques used to gain performance, and evaluate it experimentally in Section~\ref{sec:benchmarks}. Finally, we draw the conclusions and present perspectives for future work in Section~\ref{sec:conclusions}.

\section{Approximate Reachability of Affine Systems}
\label{sec:preliminaries}
%
In this section, we recall the state-of-the art in approximating the reachable set of an affine dynamical system.
\subsection{Preliminaries}

Let us introduce some notation. Let $\mathbb{I}_n$ be the identity matrix of dimension \nbyn.
For $p \geq 1$, the $p$-norm of an $n$-dimensional vector $x \in \R^n$ is denoted 
$\norm{x}_p$. The norm of a set $\X$ is $\norm{\X}_p=\max_{x \in \X} \norm{x}_p$.
Let $\Bpn$ be the unit ball of the $p$-norm in $n$ dimensions, i.e., $\Bpn = \{ x : \norm{x}_p \leq 1 \}$. 
The Minkowski sum of sets $\X$ and $\Y$ is $\X \oplus \Y := \{ x + y : x \in \X \text{ and } y \in \Y \}.$ Their Cartesian product, $\X\times \Y$, is the set of ordered pairs $(x, y)$, with $x \in \X$ and $y\in \Y$. The origin in $\R^n$ is written \norigin{n}. There is a relation between products of sets and Minkowski sum: if $\X \subseteq \Reals^n$ and $\Y \subseteq \Reals^m$, then
$\X \times \Y = (\X \times \{  \norigin{m} \} ) \oplus (\{ \norigin{n} \} \times \Y).$
The convex hull operator is written $\CH$. Let $\boxdot(\cdot)$ be the symmetric interval hull operator, defined for any $\X \subset \R^n$ as the $n$-th fold Cartesian product of the intervals $[-\vert \bar{x}_i\vert, \vert \bar{x}_i \vert]$ for all $i=1,\ldots,n$, where $\vert \bar{x}_i \vert := \sup\{\vert x_i \vert: x \in \X\}$.

\subsection{Polyhedral Approximation of a Convex Set}
\label{sec:set_approx}

We recall some basic notions for approximating convex sets. 
Let $\X \subset \Reals^n$ be a compact convex set. The \emph{support function} of $\X$ is the function $\rho_\X : \Reals^n\to \Reals$,
\begin{equation}\label{eq:support_function}
\rho_\X(\ell) := \max\limits_{x \in \X} \ell^\transp x.
\end{equation}
The farthest points of $\X$ in the direction $\ell$ 
are the \emph{support vectors}
\begin{equation}\label{eq:support_vector}
\sigma_\X(\ell) := \left\{ x \in \X : \ell^\transp x  = \rho_{\X}(\ell)  \right\}.
\end{equation}
When we speak of \emph{the} support vector, we mean the choice of any support vector in~\eqref{eq:support_vector}.
The projection of a set into a low dimensional space (a special case of $M \X$) can be conveniently evaluated using support functions, since $\sigma_{M\X}(\ell) = \sigma_\X(M^\transp\ell)$.
Given directions $\ell_1,\ldots,\ell_m$, a tight overapproximation of $\X$ is the \emph{outer polyhedron} given by the constraints 
\begin{equation}\label{eq:outerpoly}
\bigwedge_i \ell_i^\transp x \leq \rho_\X(\ell_i).
\end{equation}
For instance, a bounding box involves evaluating the support function in $2n$ directions. More precise approximations can be obtained by adding directions.
To quantify this, we use the following distance measure.
 A set $\Xhat$ is within Hausdorff distance $\varepsilon$ of $\X$ if and only if 
\begin{equation}\label{eq:hausdorff_distance}
\Xhat \subseteq \X \oplus \varepsilon\Bpn \text{ and } \X \subseteq \Xhat \oplus \varepsilon\Bpn.
\end{equation}
The infimum $\varepsilon \geq 0$ that satisfies~\eqref{eq:hausdorff_distance} is called the Hausdorff distance between $\X$ and $\Xhat$ with respect to the $p$-norm, and is denoted $\hausdorff^p\bigl(\X,\Xhat\bigr)$.
Another useful characterization of the Hausdorff distance is the following. Let $\X, \Y \subset \Reals^n$ be polytopes. Then
\begin{equation}
d^p_H(\X, \Y) = \max_{\ell \in \Bpn} |\rho_{\Y}(\ell) - \rho_{\X}(\ell)|. \label{eq:dpH_suppFun}
\end{equation}
In the special case $\X \subseteq \Y$, the absolute value can be removed.

By adding directions using Lotov's method~\cite{lotov2008modified}, the outer polyhedron in  \eqref{eq:outerpoly} is within Hausdorff distance $\varepsilon \norm{X}_p$ for $\mathcal{O}(\sfrac{1}{\varepsilon^{n-1}})$ directions, and this bound is optimal. It follows that accurate outer polyhedral approximations are possible only in low dimensions. For $n=2$, the bound can be lowered to  $\mathcal{O}(\sfrac{1}{\sqrt{\varepsilon}})$ directions, which is particularly efficient and the reason why we chose to decompose the system into subsystems of dimension 2.

\subsection{Trajectory, Reach Set, and Reach Tube} \label{ssec:trajectories_reachsets_reachtubes}

\begin{figure}
	\begin{center}
		\begin{tikzpicture}[
		domain=0:3.5,
		rt/.style={color=black},
		rtfill/.style={fill=orange},
		oa/.style={fill=yellow!40!white},
		init/.style={fill=green!80!black},
		bad/.style={fill=red}
	]
	%
	\draw[->] (-0.1,-0.5) -- (3.7,-0.5) node[right] {$t$};
	\draw[->] (0,-0.6) -- (0,2.3) node[above] {$x(t)$};
	\draw[oa] (0,-0.1) -- (0,0.9) -- (0.5,0.9) -- (0.5,-0.1) -- cycle;
	\draw[oa] (0.5,0.45) -- (0.5,1.35) -- (1,1.35) -- (1,0.45) -- cycle;
	\draw[oa] (1,0.8) -- (1,1.55) -- (1.5,1.55) -- (1.5,0.8) -- cycle;
	\draw[oa] (1.5,0.9) -- (1.5,1.55) -- (2,1.55) -- (2,0.9) -- cycle;
	\draw[oa] (2,0.55) -- (2,1.4) -- (2.5,1.4) -- (2.5,0.55) -- cycle;
	\draw[oa] (2.5,0.1) -- (2.5,1.05) -- (3,1.05) -- (3,0.1) -- cycle;
	\draw[oa] (3,-0.35) -- (3,0.5) -- (3.5,0.5) -- (3.5,-0.35) -- cycle;
	\draw[rt,name path=upper] plot (\x,{1.2 * sin(\x r) + 0.3});
	\draw[rt,name path=lower] plot (\x,{sin(\x r)});
	\tikzfillbetween[of=upper and lower,on layer=main]{rtfill}
	\draw[init] (-0.03,0) rectangle (0.03,0.3);
	\draw[bad] (0,1.8) rectangle (3.5,2.0);
	\draw (0,-0.1) -- (0,0.9) -- (0.5,0.9) -- (0.5,-0.1) -- cycle;
	\draw (0.5,0.45) -- (0.5,1.35) -- (1,1.35) -- (1,0.45) -- cycle;
	\draw (1,0.8) -- (1,1.55) -- (1.5,1.55) -- (1.5,0.8) -- cycle;
	\draw (1.5,0.9) -- (1.5,1.55) -- (2,1.55) -- (2,0.9) -- cycle;
	\draw (2,0.55) -- (2,1.4) -- (2.5,1.4) -- (2.5,0.55) -- cycle;
	\draw (2.5,0.1) -- (2.5,1.05) -- (3,1.05) -- (3,0.1) -- cycle;
	\draw (3,-0.35) -- (3,0.5) -- (3.5,0.5) -- (3.5,-0.35) -- cycle;
\end{tikzpicture}
	\end{center}
	\caption{Illustration of a reach tube (orange) with set of initial states (green) and an approximation (yellow) that shows absence of error states (red).
	}
	\label{fig:reach_tube}
\end{figure}
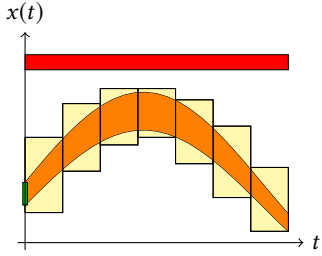

A \emph{trajectory} of the affine ODE with time-varying inputs~\eqref{eq:continuous_system} is the unique solution $x_{x_0,u}(t): [0, T] \rightarrow \Reals^n$, for a given initial condition $x_0$ at time $t=0$, and a given input signal $u$,
\begin{equation}\label{eq:pwa_analytic}
x_{x_0,u}(t) = e^{At}x_0 + \int_0^t e^{A(t-s)}u(s)\,ds,
\end{equation}
where we map $Bu(t)$ to $u(t)$ without loss of generality. Here $T$ is the time horizon, which is considered to be finite in this paper.
Given a set of initial states $\X_0$ and an input signal $u$, the \emph{reach set} at time $t$ is
$
\Reach(\X_0,u,t) := \{ x_{x_0,u}(t) : x_0 \in \X_0 \}.
$
This extends to a family of solutions as
\begin{equation}\label{eq:def_reach_ndetinp}
\Reach(\X_0, \U,t) = \bigcup \bigl\{\Reach(\X_0,u,t) : u(s) \in \U(s)\,, \forall~s \in [0, t] \bigr\}.\hspace*{-1mm}
\end{equation}
The \emph{reach tube} for a given time interval $[t_1, t_2] \subseteq [0, T]$ is the set
\begin{equation}\label{eq:reach_tube}
\Reach(\X_0,\U, [t_1, t_2]) := \bigcup_{t_1 \leq t \leq t_2} \Reach(\X_0,\U,t).
\end{equation}
In general, the reach tube can be computed only approximately.
An example reach tube and an overapproximation using boxes is shown in \fig{fig:reach_tube}.
In the next section we discuss how to compute such an overapproximation of the reach tube.

\subsection{Approximation Model}
\label{sec:approx_model}

The standard numerical approach for the reachability problem is to reduce it to computing a finite sequence of sets, $\{\X(k)\}_{k=0}^N$, that overapproximates the exact reach tube~\eqref{eq:reach_tube}. 
We assume a given constant time step size $\delta > 0$ over the time horizon $T=N\delta$, where $N$ is the number of time steps. 
With respect to the inputs, we assume that the time-varying function $\U(\cdot)$ from \sect{ssec:trajectories_reachsets_reachtubes} is piecewise constant, i.e., we consider a possibly time-varying discrete sequence $\{\U(k)\}_k$ for all $k=0,1,\ldots, N$.

In the dense-time case, one is interested in covering all possible trajectories of the given continuous system. In the discrete-time case, the reach tube of the discretized system is only covered at discrete time steps, but not necessarily between time steps. In either case, starting from the system~\eqref{eq:continuous_system}-\eqref{eq:output_system}, we can reduce the reachability problem to the general recurrence~\eqref{eq:discrete_affine_recurrence}, with suitably transformed initial states and nondeterministic input. These reductions can be found in previous works~\cite{LeGuernic2010250, frehse2011spaceex} and we recall them below.

First, we recall the dense time case. 
All continuous trajectories are covered by the discrete approximation if
\begin{equation}\label{eq:reach_set_overapproximation}
 \Reach(\X_0, \U, [k\delta, (k+1)\delta]) \subseteq \X(k),\quad  k=0,1,\ldots, N
\end{equation}
Previous works have provided approximation models such that~\eqref{eq:reach_set_overapproximation} holds~\cite{LeGuernic09, LeGuernic2010250, frehse2011spaceex}. In particular, in~\cite[Lemma 3]{frehse2011spaceex} the authors intersect a first-order approximation of the interpolation error going forward in time from $t=0$ with one that goes backward in time from $t=\delta$. Note that this forward-backward approximation is used in \spaceex, to which we will compare our method later.
Here, we consider the forward-only approximation. 
To guarantee that the overapproximation covers the interval between time steps, the  initial set and the input sets are bloated by additive terms%
\begin{equation*}
\begin{aligned}
E_\psi(\U(k), \delta) &:= \boxdot(\Phi_2(\vert A \vert, \delta) \boxdot(A \U(k))) \\
E^+(\X_0, \delta) &:= \boxdot(\Phi_2(\vert A \vert, \delta) \boxdot(A^2  \X_0)),
\end{aligned}
\end{equation*}
where the matrices $\Phi_{1}(A, \delta)$ and $\Phi_{2}(A, \delta)$ are defined via
\begin{equation*}
\Phi_1(A, \delta) := \sum\limits_{i=0}^\infty \dfrac{\delta^{i+1}}{(i+1)!}A^i,\hspace*{2mm} \Phi_2(A, \delta) := \sum\limits_{i=0}^\infty \dfrac{\delta^{i+2}}{(i+2)!}A^i.
\end{equation*}
The required transformations for dense time are:
\begin{flalign}\label{eq:approximation_model_dense_time}
\left\{\
\begin{aligned}
\Phi &\leftarrow e^{A \delta} \\
\X(0) &\leftarrow  \CH\big( \X_0, \Phi \X_0 \!\oplus\! \delta \U(0)
\!\oplus\! E_\psi(\U(0),\delta)\!\oplus\! E^+(\X_0, \delta)\big)\\
\V(k) &\leftarrow  \delta \U(k) \oplus E_\psi(\U(k), \delta),\quad \forall~k = 0,1,\ldots, N \hspace*{-5mm}
\end{aligned}\hspace*{-1em}
\right. &&
\end{flalign}
For discrete time reachability the transformations are:
\begin{flalign}\label{eq:approximation_model_discrete_time}
\left\{\
\begin{aligned}
\Phi &\leftarrow e^{A \delta} \\
\X(0) &\leftarrow  \X_0\\
\V(k) &\leftarrow  \Phi_1(A, \delta) \U(k),\quad \forall~k = 0,1,\ldots, N
\end{aligned}
\right. &&
\end{flalign}
Note that there is no bloating of the initial states, and that the inputs are assumed to remain constant between sampling times.

The cost of solving the general recurrence~\eqref{eq:discrete_affine_recurrence} with either the data~\eqref{eq:approximation_model_dense_time} or~\eqref{eq:approximation_model_discrete_time}, to compute an approximation of the reach set or the reach tube, increases superlinearly with the dimension of the system and the desired approximation error. In the rest of this paper, we will consider a decomposition of the system to reduce the computational cost.

\section{Decomposition}
\label{sec:decomposition}
%
In this section, we present a novel approach for solving the general recurrence~\eqref{eq:discrete_affine_recurrence} using block decompositions.

\subsection{Cartesian Decomposition}

From now on, let $\X \subset \Reals^n$ be a compact and convex set. To simplify the discussion, we assume $n$ to be even. We characterize the decomposition of $\X$ into $b:=n/2$ sets of dimension two as follows.
Let $\projmat_i$ be the \emph{projection matrix} that maps a vector $x \in \Reals^n$ to its coordinates in the $i$-th block,
$x_i = \projmat_i x$.
The \emph{Cartesian decomposition} of $\X$ is the set
$$\exactdecompose(\X) := \projmat_1 \X \times \cdots \times \projmat_b \X.$$
We call a set \emph{decomposed} if it is identical to its Cartesian decomposition. For instance, the symmetric interval hull $\boxdot(\X)$ is a decomposed set, since it is the Cartesian product of one-dimensional sets, i.e., intervals.
Throughout the paper, we will highlight decomposed sets with the symbol $\hat\cdot$ (as in $\Xhat, \Yhat$).
Note that decomposition distributes over Minkowski sum:
\begin{align}
\exactdecompose(\X \oplus \Y) &= \exactdecompose(\X) \oplus \exactdecompose(\Y) \label{eq:decomp_minkowski}.
\end{align}

If $\X$ is a polyhedron in constraint form, the projections can be very costly to compute, which amounts to quantifier elimination. 
However, using the methods in \sect{sec:set_approx}, we can efficiently compute an overapproximation. The overapproximation can be coarse, e.g., a bounding box, or $\varepsilon$-close in the Hausdorff norm for a given value of $\varepsilon$.
Since the choice of approximation is of no particular importance to the remainder of the paper, we simply assume an operator $$\Xhat_1\times \cdots \times \Xhat_b =\decompose(\X)$$ that overapproximates the Cartesian decomposition with a decomposed set $\Xhat_1\times \cdots \times \Xhat_b$ such that 
$\exactdecompose(\X) \subseteq \decompose(\X).$

\subsection{Decomposing an Affine Map} \label{ssec:decomposing_an_affine_map}

\begin{figure}[t]
	\begin{center}
		\includegraphics[scale=0.4, trim=0 20 0 0]{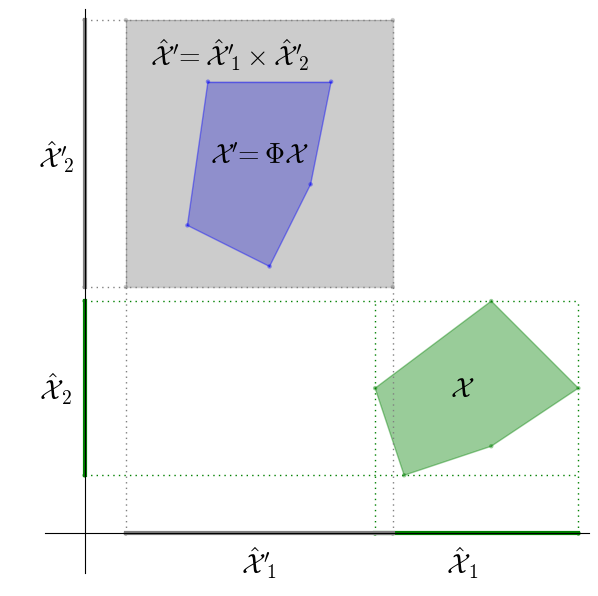}
	\end{center}
	\caption{The Cartesian decomposition of $\X$ (green) in two blocks of size one is the set $\decompose(\X)=\Xhat_1\times \Xhat_2$. The decomposed image of the map $\X'=\Phi\X$ (blue) is the set $\Xhat'$ (gray) obtained by the application of Eq.~\eqref{eq:linear_map_decomposed} for each block $\Xhat_j'$.
	}
	\label{fig:box_approximation}
\end{figure}

Suppose that a compact and convex set $\set{V}\subset \Reals^n$ is given, and let $\Phi$ be a real \nbyn matrix. Consider the $n$-dimensional \emph{affine map}
\begin{equation}\label{eq:linear_map_original}
\X' = \Phi \X \oplus \V = \begin{pmatrix}
\Phi_{11} & \cdots & \Phi_{1b} \\
\vdots & \ddots & \vdots \\
\Phi_{b1} & \cdots & \Phi_{bb} 
\end{pmatrix}
\X \oplus \V,
\end{equation}
where~$\Phi_{ij}$ denotes the \nbyn[2] submatrix of~$\Phi$ in row~$i$ and column~$j$, counting from top to bottom and from left to right. We call such a submatrix a \emph{block}, and $[\Phi_{i1} \Phi_{i2} \cdots \Phi_{ib}]$ a \emph{row-block}. We assume without loss of generality that $n$ is even. Hence, every row-block of~$\Phi$ consists of $b := n/2$ two-dimensional blocks. 

The \emph{decomposed image} of the map~\eqref{eq:linear_map_original} is obtained in two steps. First, we transform the full-dimensional sets $\X$ and $\V$ into Cartesian products of two-dimensional sets $\Xhat_1, \ldots, \Xhat_b$, using the operator $\decompose$  described in the previous section.
Second, we construct the two-dimensional sets
\begin{equation}\label{eq:linear_map_decomposed}
\Xhat_i' := \bigoplus_{j=1}^b \Phi_{ij} \Xhat_j \oplus \Vhat_i,\quad \forall~i = 1,\ldots, b.
\end{equation}
We call the Cartesian product $\Xhat' = \Xhat'_1 \times \cdots \times \Xhat'_b$ the decomposed image of~\eqref{eq:linear_map_original}. For each $i$, $\Xhat'_i$ only depends on the $i$-th row-block of $\Phi$. We illustrate in \fig{fig:box_approximation} the decomposed image of an affine map over a polygon.

We now compare the cost of~\eqref{eq:linear_map_original} and~\eqref{eq:linear_map_decomposed}. Let us denote the cost of computing the image of an \nbyn linear map by $C_\odot(n, m)$ and the cost of computing the Minkowski sum of two $n$-dimensional sets by $C_\oplus(n, m)$, where $m$ is a parameter that depends on the set representation. The asymptotic complexity of performing the above operations for common set representations is shown in Table~\ref{table:costSetsOps}; we refer to~\cite{Fukuda04,Monniaux10quantifiers,Girard05,GirardLG08} for further details. Since one Minkowski sum and one linear map are involved in~\eqref{eq:linear_map_original}, we have that 
$$
\text{Cost~\eqref{eq:linear_map_original}} \in \Oof{C_\odot(2b, m) + C_\oplus(2b, m)}.
$$
On the other hand,
the aggregated cost for the $i$-th  block in~\eqref{eq:linear_map_decomposed} is $b C_\oplus(2, m') + b C_\odot(2, m')$, where $m'$ is the parameter for complexity ($m$) in two dimensions. The total cost is thus 
$$
\text{Cost~\eqref{eq:linear_map_decomposed}} \in \Oof{b^2 C_\odot(2, m') + b^2 C_\oplus(2, m')}.
$$
Whenever $C_\odot$ and $C_\oplus$ depend at least quadratically on the dimension, and since $m' \ll m$, the cost of the decomposed image~\eqref{eq:linear_map_decomposed} is asymptotically smaller than the cost of the non-decomposed image~\eqref{eq:linear_map_original}.

\def\tabspace{\;}
\begin{table}[t]
	\caption{
	Complexity of set operations involved in the affine map computation by decomposition.
	}
	\label{table:costSetsOps}
	\parbox{\columnwidth}{
	\centering
	\begin{tabular}{@{} r @{\tabspace}  @{\tabspace} c @{\tabspace}  @{\tabspace} c @{\tabspace}  @{\tabspace} c @{\tabspace}  @{\tabspace} c @{}}
		\toprule
		& \multicolumn{2}{c @{\tabspace}  @{\tabspace}}{polyhedra} & zonotopes & supp. fun. \\
		\cline{2-3}
		& $m$ constraints & $m$ vertices & $m$ generat. & $m$ direct. \\
		\midrule
		&&&& \\[-3mm]
		$C_\odot(n, m)$ & $\Oof{mn^2+n^3}$ & $\Oof{mn^2}$ & $\Oof{m n^2}$ & $\Oof{m n^2 \mathcal{L}}$ \\[1mm]
		$C_\oplus(n, m)$ & $\Oof{2^n}$ & $\Oof{m^2n}$ & $\Oof{n}$ & $\Oof{m \mathcal{L}}$ \\
		\bottomrule
	\end{tabular}}
\begin{tablenotes}
	\footnotesize
	$\mathcal{L}$ is the cost of evaluating the support function of~$\set{X}$.
	For polyhedra in constraint representation we assume that $\Phi$ is invertible; otherwise the complexity is $\Oof{m^n}$.
	Note that~$m$ is not comparable between different representations.
\end{tablenotes}
\end{table}

\subsection{Decomposing an Affine Recurrence} \label{ssec:decomposing_an_affine_recurrence}

Let us reconsider the affine recurrence in Eq.~\eqref{eq:discrete_affine_recurrence}.
We can rewrite it into~$b$ row-blocks, as in \sect{ssec:decomposing_an_affine_map}, with given compact and convex sequence $\{\V(k)\}_k \subset \R^n$ for $k\geq 0$, and an initial set $\X(0)$:
\begin{equation}\label{eq:linear_map}
\begin{aligned}
\X(k+1)
&= \begin{pmatrix}
\Phi_{11} & \cdots & \Phi_{1b} \\
\vdots & \ddots & \vdots \\
\Phi_{b1} & \cdots & \Phi_{bb} 
\end{pmatrix} 
\X(k) \oplus \V(k).
\end{aligned}
\end{equation}
In this recurrence, the approximation error of the $k$-th step is propagated, and possibly amplified, in step $k+1$. This can be partly avoided by using a non-recursive form ~\cite{GirardLGM06}. We present two scenarios, which differ in whether the sequence of input sets is constant or not.
Let $\Phi^k_{ij}$ be the submatrix of $\Phi^k$ corresponding to the indices of the submatrix $\Phi_{ij}$ of~$\Phi$.

\paragraph{Constant input sets}

Assuming that the sets $\set{V}$ do not depend on~$k$, the non-recurrent form of~\eqref{eq:linear_map} is:
\begin{flalign*}
\left\{\ 
\begin{aligned}
\X(k) &= \Phi^k  \X(0) \oplus \W(k) \\
\W(k+1) &=  \W(k) \oplus \Phi^k \V, \quad \W(0) :=  \{\norigin{n}\}.
\end{aligned}
\right. &&
\end{flalign*}
The decomposed map, for $i=1,\ldots, b$, is:
\begin{flalign}\label{eq:constant_nonrecurrent_decomposed}
\left\{\
\begin{aligned}
\Xhat_i(k) &= \bigoplus_{j=1}^b \Phi_{ij}^k  \set{\hat X}_j(0) \oplus \What_i(k) \\
\What_i(k+1) &= \What_i(k) \oplus [\Phi^k_{i1}\cdots \Phi^k_{ib}] \V, \quad \What_i(0) :=  \{\norigin{2}\}. \hspace*{-5mm}
\end{aligned}
\right. \hspace*{-2mm}&&
\end{flalign}
Note that the set $[ \Phi_{i1}^k \cdots  \Phi_{ib}^k ] \V$ in~\eqref{eq:constant_nonrecurrent_decomposed}
is of low dimension and corresponds to the $i$-th block. 

\paragraph{Time-varying input sets}

Assuming that the sequence of inputs depends on~$k$, the non-recurrent form of~\eqref{eq:linear_map} is:
\begin{flalign*}
\left\{\ 
\begin{aligned}
\X(k) &= \Phi^k  \X(0) \oplus \W(k) \\
\W(k+1) &=  \Phi\W(k) \oplus  \V(k), \quad \W(0) :=  \{\norigin{n}\}.
\end{aligned}
\right. &&
\end{flalign*}
The decomposed map, for $i=1,\ldots, b$, is:
\begin{flalign}\label{eq:variant_nonrecurrent_decomposed}
\left\{\ 
\begin{aligned}
\Xhat_i(k) &= \bigoplus_{j=1}^b \Phi_{ij}^k  \set{\hat X}_j(0) \oplus \What_i(k)\\
\What_i(k+1) &= \bigoplus_{j=1}^b \Phi_{ij} \What_j(k) \oplus \Vhat_i, \quad \What_i(0) := \{\norigin{2}\}. \hspace*{-5mm}
\end{aligned}
\right. &&
\end{flalign}

\section{Approximation error}
\label{ssec:approximation_error}
%
In general, the reduction in the computational cost of the decomposed image comes at the price of an approximation error for $\Xhat_i(k)$. We discuss the two sources of this error.

The first one is due to the decomposition of the initial states. For discrete time reachability, the initial set $\X_0$ remains unchanged under the transformations~\eqref{eq:approximation_model_discrete_time}. In practice, $\X_0$ often has the shape of a hyperrectangle, and hence there is no approximation error. However, for dense time reachability the transformations~\eqref{eq:approximation_model_dense_time} do not preserve an initially decomposed set, and \decompose~invariably introduces an approximation error. If the constraints on $\X(0)$ are known, an upper bound on the Hausdorff distance $d_H^p(\X(0), \Xhat(0))$ can be obtained using support functions~\cite{lotov2008modified}.

The second source of the approximation error is the step-wise decomposition of the inputs. This can be either a linear combination with respect to a row-block, as in~\eqref{eq:constant_nonrecurrent_decomposed}, or a single block as in~\eqref{eq:variant_nonrecurrent_decomposed}. For a stable matrix $\Phi$, in either case, the error propagated to $\What_i(k+1)$ goes to zero for $k \to \infty$. 
In the rest of the section, we discuss these errors in more detail.

\subsection{Error of a Decomposed Affine Map}
We now turn to the question how big the decomposition error is in the decomposed affine map~\eqref{eq:linear_map_decomposed} compared to~\eqref{eq:linear_map_original}.
To simplify the discussion, we omit $\V$ without loss of generality, since we can rephrase~\eqref{eq:linear_map_original} with an augmented state space where $\X^* \gets \X \times \V$ and $\Phi^* \gets [ \Phi\ I ]$. Then $\X' = \Phi^* \X^*$.

We proceed in two steps. First, we bound the error for a set that is already decomposed. Then we bound the distance between the image of the decomposed and the original set. The total error  follows from a triangle inequality.
Let $\set{\Xhat}$ be a decomposed set,
and let $\set{\bar X}'$ be the image of $\Xhat$ under the linear map $$\set{\bar X}'=\Phi\set{\Xhat}$$
 and let $\set{\hat X}'$  the image of $\Xhat$ under the decomposed map~\eqref{eq:linear_map_decomposed}. 
\begin{proposition}\label{prop:HDistDecompMap}
$\set{\bar X}' \subseteq \set{\hat X}'$ and
\begin{equation}\label{eq:error_crossprod}
\hspace*{-1mm}
d_H^p(\set{\bar X}', \set{\hat X}') =  \max_{\Vert d \Vert_p \leq 1} \sum_{i,j} \rho_{ \set{\hat X}_j}(\Phi_{ij}^T d_i) - \rho_{ \set{\hat  X}_j}\left(\sum_k \Phi_{kj}^T d_k \right)
\hspace*{-1mm}
\end{equation}
where the $\max$ is taken over $d = d_1\times \cdots \times d_b$ in the unit ball of the $p$-norm, and $\Phi_{ij}^T := (\Phi_{ij})^T$.
\end{proposition}
\begin{corollary}
	If only one $\Phi_{ij}$ per column is nonzero, then the error is zero.
	A special case of such a matrix is the (real) Jordan form if all eigenvalues have multiplicity 1.
\end{corollary}

We can simplify this bound to clarify the relationship between the error and the norm of the matrix blocks $\Phi_{ij}$.
 To quantify the error associated to the $i$-th block, let us introduce the number $\Delta_i$ to be the diameter of $\Xhat_i$, i.e., the smallest number such that for any $d$,
$$\rho_{\Xhat_i}(d)+\rho_{\Xhat_i}(-d) \leq  \norm{d}_{\tfrac{p}{p-1}} \Delta_i.$$
For instance, if $\Xhat_i$ is an interval hull, then $\Delta_i$ is the width of the largest interval.
Then the error bound is a weighted sum of the diameter of state sets:
\begin{proposition}\label{prop:HDistDecompMap_simplified3}
For $j=1,\ldots,b$, let $q_j := \argmax_i \norm{\Phi_{ij}}_p$ (the index of the block with the largest matrix norm in the \mbox{$j$-th} column-block), so that
 $\alpha_j := \max_{i \neq q_j} \norm{\Phi_{ij}}_p$ is the second largest matrix norm in the $j$-th column-block. The error of the decomposed map is
\begin{equation}\label{eq:HDistDecompMap_simplified3_ppetit}
\hausdorff^p\bigl(\set{\bar X'},\set{\hat X}'\bigr) \leq 
(b-1) \sum_{j=1}^b \alpha_j \Delta_j \leq \frac{n}{2} \alpha_\mathrm{max} \Delta_\mathrm{sum},
\end{equation}
where $\alpha_\mathrm{max} := \max_j \alpha_j$ and $\Delta_\mathrm{sum} :=\sum_{j=1}^b \Delta_j$.
\end{proposition}
We can interpret Prop.~\ref{prop:HDistDecompMap_simplified3} as a generalization of the observation that the approximation error is small if the off-diagonal entries of $\Phi$ are small. 

We now come to the second step. We compare the image $\X'=\Phi \X$ with the decomposed image $\Xhat'$, including both the decomposition error from $\X$ to $\Xhat$ and the error introduced by the decomposed map~\eqref{eq:linear_map_decomposed}. We use a simple lemma:
\begin{lemma}\label{lem:distmap}
Let $\X'=\Phi\set{\X}$ and $\set{\bar X}'=\Phi\set{\Xhat}$, where $\set{\X} \subseteq \Xhat$. The distance between the images is
$\hausdorff^p\bigl(\set{X'},\set{\bar X'}\bigr) \leq 
\norm{\Phi}_p \hausdorff^p\bigl(\set{\X},\set{\Xhat}\bigr).
$
\end{lemma}
Combining Lemma~\ref{lem:distmap} with Prop.~\ref{prop:HDistDecompMap_simplified3} and the triangle inequality 
$\hausdorff^p\bigl(\set{X'},\set{\hat X'}\bigr) \leq \hausdorff^p\bigl(\set{X'},\set{\bar X'}\bigr) + \hausdorff^p\bigl(\set{\bar X'},\set{\hat X'}\bigr)$,
we get the following total error bound on the decomposed image computation:
\begin{proposition}\label{prop:total_aff_err_bound}
$
\hausdorff^p\bigl(\set{ X'},\set{\hat X}'\bigr) \leq 
(b-1) \sum_{j=1}^b \alpha_j \Delta_j + \norm{\Phi}_p \hausdorff^p\bigl(\set{\X},\set{\Xhat}\bigr).
$
\end{proposition}
The above bound gives us an idea about the error of the decomposed affine map, without having to do any high-dimensional set computations.
We now apply it to affine recurrences.

\subsection{Error of a Decomposed Affine Recurrence}
For any $\Phi$, there exist constants $K_\Phi$ and $\alpha_\Phi$ such that
$$\norm{\Phi^k}_p \leq K_\Phi \alpha_\Phi^k,\qquad k\geq 0.$$

If $\Phi=e^{A\delta}$, one choice is $\alpha_\Phi=e^{\lambda \delta}$ with $\lambda$ the spectral abscissa (largest real part of any eigenvalue of $A$), although it may not be possible to compute the corresponding $K_\Phi$ efficiently. In this case, $\alpha_\Phi \leq 1$ if the system is stable. Another choice is to let $\alpha_\Phi=e^{\mu \delta}$, with $\mu$ the logarithmic norm of $A$ and $K_\Phi=1$. In this case, $\alpha_\Phi$ may be larger than~$1$ even for stable systems. Note that in both cases $\alpha_\Phi \to 1$ as $\delta \to 0$. For concreteness we continue with the first formulation in the remaining section.

For constant inputs sets,~\eqref{eq:constant_nonrecurrent_decomposed} is a linear map of the decomposed initial states $\Xhat(0)$ plus a decomposed input $\hat\W(k)$, which is itself obtained from a sequence of decomposed linear maps.
Applying Prop.~\ref{prop:total_aff_err_bound} gives the following result.

\begin{proposition}\label{prop:aff_rec_error}
Let the decomposition error of the initial states $\X(0)$ be bounded by $\varepsilon^x \geq \hausdorff^p\bigl(\X(0),\set{\hat X}(0)\bigr),$
and let the decomposition error of $\set{V}$ be bounded by
$\varepsilon^v \geq \hausdorff^p\bigl(\set{ V},\set{\hat V}\bigr).$
Let $\Delta^x_j$ be the diameter of $\Xhat_j(0)$, and $\Delta^x_\mathrm{sum}=\sum_{j=1}^b \Delta^x_j$.
Let $\Delta_j^v$ be the diameter of $\set{\hat V}_j$, and $\Delta^v_\mathrm{sum}=\sum_{j=1}^b \Delta_j^v$.
Then the approximation error due to decomposition, at step $k$, is bounded by
\begin{equation*}\label{eq:decomp_aff_rec_err_bnd}
\begin{split}
\hausdorff^p\bigl(\set{\Xhat}(k),\X(k)\bigr) &\leq 
 K_\Phi \Bigl( \alpha_\Phi^k   \bigl(b \Delta^x_\mathrm{sum} + \varepsilon^x \bigr) \\
&\quad + \bigl(b\Delta^v_\mathrm{sum}+\varepsilon^v\bigr) \alpha_\Phi \frac{1-\alpha_\Phi^{k-1}}{1-\alpha_\Phi}\Bigr) + \varepsilon^v.
\end{split}
\end{equation*}
If $\alpha_\Phi<1$ (stable system), the error is bounded for all~$k$ by
\begin{align*}
\hausdorff^p\bigl(\set{\Xhat}(k),\X(k)\bigr) &\leq K_\Phi   \Bigl(  b \Delta^x_\mathrm{sum} + \varepsilon^x  +  \bigl( b \Delta^v_\mathrm{sum}+\varepsilon^v\bigr)  \frac{\alpha_\Phi}{1-\alpha_\Phi} \Bigr) + \varepsilon^v.
\end{align*}
\end{proposition}
In conclusion, the approximation error is linear in the width of the initial states and the inputs, and in the decomposition errors of the initial states and the input sets. For unstable systems, or time steps not large enough, the input set can become the dominating source of error, e.g., considering  cases with $\alpha_\Phi>\frac{1}{2}$.

\subsection{Error of a Decomposed Reach Tube Approximation}
The decomposed reach tube approximation consists of the affine recurrence~\eqref{eq:constant_nonrecurrent_decomposed},  with suitable sets~$\X(0)$ and~$\set{V}$. The error bound follows from Prop.~\ref{prop:aff_rec_error} and the decomposition errors for $\X(0)$ and $\set{V}$.

In the discrete time case~\eqref{eq:approximation_model_discrete_time}, the initial states $\X(0)$ of the affine recurrence ~\eqref{eq:constant_nonrecurrent_decomposed}
are identical to the initial states $\X_0$ of the model, so their decomposition error is
$$\varepsilon^x =  \hausdorff^p\bigl(\X_0,\set{\hat X}_0\bigr).$$
However, $\V =  \Phi_1(A, \delta) \U$.
Let $\set{\hat U} = \exactdecompose(\U)$. By Lemma~\ref{lem:distmap} we get  
$$\varepsilon^v = \norm{ \Phi_1(A, \delta) }_p  \hausdorff^p\bigl(\set{U},\set{\hat U}\bigr).$$

In the dense time case~\eqref{eq:approximation_model_dense_time}, the initial states of the affine recurrence~\eqref{eq:constant_nonrecurrent_decomposed} are $\X(0)= \CH\big( \X_0, \Phi \X_0 \oplus \delta \U
 \oplus\, E_\psi(\U,\delta)\oplus E^+(\X_0, \delta)\big)$, and
 $\V =  \delta \U \oplus E_\psi(\U, \delta)$.
Recall from~\eqref{eq:decomp_minkowski} that decomposition distributes over Minkowski sum.
We get
$$\varepsilon^v = \delta \hausdorff^p\bigl(\set{U},\set{\hat U}\bigr).$$ 
 The decomposition error for the initial states is more complex and harder to estimate. 

We now consider the idealized case where 
 the system is stable with $\alpha_\Phi=e^{-\lambda \delta}$, $\lambda>0$, for an infinitesimal time step $\delta \to 0$. Then $\alpha_\Phi \to 1-\lambda\delta$ and $\frac{\alpha_\Phi}{1-\alpha_\Phi} \to \frac{1}{\lambda\delta}$, so that the decomposition error due to the inputs does not go to zero in Prop.~\ref{prop:aff_rec_error}.
Let $\Delta_{\X_0}$, $\Delta_{\U}$ be the sum of the diameters of decomposed sets of $\X_0$ and $\U$. Let
$\varepsilon^x_0 =  \hausdorff^p\bigl(\X_0,\set{\hat X}_0\bigr)$ and $\varepsilon^v_0 = \hausdorff^p\bigl(\set{U},\set{\hat U}\bigr)$.
For both the discrete time and the dense time case, $\varepsilon^x \to \varepsilon^x_0$, $ \Delta^x_\mathrm{sum} \to \Delta_{\X_0}$, $\Delta^v_\mathrm{sum} \to \delta \Delta_{\U}$  and $\varepsilon^v \to \delta \varepsilon^v_0$.
Then
Prop.~\ref{prop:aff_rec_error} gives a nonzero upper bound 
$$
\hausdorff^p\bigl(\set{\Xhat}(k),\X(k)\bigr) \leq K_\Phi   \Bigl(  b \Delta_{\X_0}  + \varepsilon_0^x + \bigl( b \Delta_{\U}+\varepsilon^v_0 \bigr)    \frac{1}{\lambda} \Bigr) + \Oof{\delta}.
$$
This indicates that a small time step may be problematic for systems with large time constants (small $\lambda$).

\section{Algorithm \& implementation}
\label{sec:implementation}
%
In this section, we rephrase the decomposition method outlined in the previous sections in a more algorithmic view and discuss some crucial details for our implementation in \emph{Julia}~\cite{bezanson2017julia}.
In a nutshell, 
given an LTI system in the form~\eqref{eq:continuous_system}-\eqref{eq:output_system}, we first apply a suitable approximation model from \sect{sec:approx_model} (\discretize). Then we execute the corresponding decomposed recurrence from \sect{ssec:decomposing_an_affine_recurrence} to compute the reach tube (\reach) or to check a safety property. Finally, we project onto output variables (\projectandplot).

We have implemented several critical performance enhancements. Some of them are only applicable to the decomposition method described in this paper, and others can be applied to non-decomposed methods as well.
We give more details below:

\paragraph{Lazy data structures} 
We use lazy (i.e., symbolic) set representations for most of the set operations, in particular for Minkowski sum, linear map, and Cartesian product.
Common sets such as hypercubes in different norms, polyhedra, and polygons each are represented by specific types. Each type has to provide a function to compute the support vector in a given direction. The operations can be nested symbolically without actually evaluating them. Then, we can compute the support vector of the (nested) lazy set on demand.

The advantage of lazy data structures is that we may save unnecessary evaluations at the cost of higher memory consumption. In practice, we use a careful balance between lazy sets and concrete sets, i.e., the nesting depth is fixed (depending on the model dimension~$n$).
The alternative to using lazy data structures is to make the representation explicit after each operation, potentially involving an overapproximation.

\paragraph{Sparsity specialization} We use specific code for sparse and dense matrices. The decomposed method only needs a lookup of the non-zero blocks to evaluate $\Xhat_i$, which is particularly relevant if~$\Phi$ and its matrix powers are very sparse. Moreover, as the linear algebra back-end we use either a BLAS-compatible library~\cite{anderson1990lapack} or a native Julia implementation for sparse matrices following Gustavson~\cite{gustavson1978two}. These optimizations have a major impact on the runtime (around one order of magnitude; see the next section).

\paragraph{Target-specific analysis} If we are only interested in tracking a handful of variables, our approach naturally supports the computation of only some of the blocks. Complexity-wise this saves us a factor of~$b$ when tracking a constant number of blocks (see \sect{sec:reachable_states_computation}).

\paragraph{Lazy matrix exponentiation} We support exponentiation techniques for large and sparse matrices (e.g., $n = 10,000$), which has a major impact on runtime performance and memory cost. Instead of computing the matrix exponential $\Phi = e^{A \delta}$ explicitly, we can evaluate the action of a matrix on an $n$-dimensional vector.
We use \texttt{Expokit.jl}~\cite{expokitJL}, a Julia implementation of 
\emph{Expokit}~\cite{sidje1998expokit}.

\paragraph{Fast 2D LPs} For non-decomposed approaches, manipulating polygons or polyhedra involves using an external linear programming (LP) back-end, possibly in high-dimensional space. The restriction to polygons allows us an efficient implementation for evaluating the support vector, as explained in more details in \sect{sec:reachable_states_computation}.

\subsection{Discretization}

The \discretize~step transforms the system $(A, \U(\cdot), \X_0)$ to its discrete counterpart $(\Phi, \V(\cdot), \X(0))$. 
Recall from \sect{sec:approx_model} that both the definition of the discretized input sequence, $\V(\cdot)$, and the discretized initial states, $\X(0)$, depend on the approximation model (dense time vs.\ discrete time).
The set transformations are performed lazily for all but the symmetric interval hull operator, while the matrix exponentiation can be either explicit or lazy.

\subsection{Reach Tube Approximation}
\label{sec:reachable_states_computation}

\begin{algorithm}[t]
	\caption{Function \reach.}
	\label{algorithm:reach}
	\KwIn{%
	$\discsys = (\Phi, \V(\cdot), \X(0))$: discrete system \\
	$N$: total number of steps \\
	\textit{blocks}: list of block indices
	}
	\vspace*{1mm}
	\KwOut{$\{\Xhat(k)\}_k$: array of 2D reach tubes}
	\vspace*{2mm}
	$\Xhat(0)$ $\gets$ $\decompose(\X(0))$\; \label{line:initial_states}
	\textit{all\_blocks} $\gets$ \code{get\_all\_block\_indices}($\dim(\Phi)$)\;
	$P$ $\gets$ $\id[\dim(\Phi)]$\;
	$Q$ $\gets$ $\Phi$\;
	$\set{\hat{V}_{\text{tmp}}}$ $\gets$ []\;
	\For{$b_i \in$ \textit{blocks}}{
		$\set{\hat{V}_{\text{tmp}}}[b_i]$ $\gets$ $\{\norigin{2}\}$\;
	}
	\For{$k = 1$ \KwTo $N-1$}{ \label{line:reach_main_loop}
		$\Xhat_{\text{tmp}}$ $\gets$ []\;
		\For{$b_i \in$ \textit{blocks}}{ \label{line:reach_computation_start}
			$\Xhat_{\text{tmp}}[b_i]$ $\gets$ $\{\norigin{2} \}$\;
			\For{$b_j \in$ \textit{all\_blocks}}{
				$\Xhat_{\text{tmp}}[b_i]$ $\gets$ $\Xhat_{\text{tmp}}[b_i] \oplus Q[b_i, b_j] \odot \Xhat(0)[b_j]$; \label{line:reach_innermost}
			}
			$\set{\hat{V}_{\text{tmp}}}[b_i]$ $\gets$ \approximate($\set{\hat{V}_{\text{tmp}}}[b_i] \oplus P[b_i, :] \odot \V(k-1))$\; \label{line:reach_inputs}
			$\Xhat_{\text{tmp}}[b_i]$ $\gets$ \approximate($\Xhat_{\text{tmp}}[b_i] \oplus \set{\hat{V}_{\text{tmp}}}[b_i])$\; \label{line:reach_states_plus_inputs}
		} \label{line:reach_computation_end}
		$\Xhat(k)$ $\gets$ $\Xhat_{\text{tmp}}$\; \label{line:reach_result}
		$P$ $\gets$ $Q$\;
		$Q$ $\gets$ $Q \cdot \Phi$\;
	}
\end{algorithm}

After we have obtained a discretized system, we use Algorithm~\ref{algorithm:reach} to compute an approximation of the reach tube.
As an additional input the algorithm receives an array of block indices (\textit{blocks}) that we are interested in.
For simplicity we assume that the elements of \textit{blocks} have the form $[j, j+1]$ for even~$j$.

The result, a reach tube for each time interval~$k$, is represented by the array $\{\Xhat(k)\}_k$.
The type of each entry $\Xhat(k)$ itself is an array of polygons representing the 2D reach tubes.
To reconstruct the full-dimensional reach tube for time interval~$k$, the result has to be interpreted as a Cartesian product, i.e., $\bigotimes_{b_i} \Xhat(k)[b_i]$.
Initially, $\Xhat(0)$ just contains the decomposed initial (line~\ref{line:initial_states}).

The list \textit{all\_blocks} just consists of all the 2D block indices.
We maintain the matrix $Q$ to be the matrix $\Phi$ raised to the power of~$k$, i.e., $Q = \Phi^k$ at step $k$; similarly, $P = \Phi^{k-1}$.
For clarity we use $Q[b_i, b_j]$ instead of $Q_{ij}$ as in \sect{sec:decomposition}, and similarly, $P[b_i, :]$ denotes the whole row-block $b_i$.

The main loop starting in line~\ref{line:reach_main_loop} computes the reach tubes for each~$k$.
We write $\oplus$ and $\odot$ to denote lazy set representation of Minkowski sum and linear map, respectively.
The array $\Xhat_{\text{tmp}}$ is filled with two-dimensional reach tubes in the inner loop (lines~\ref{line:reach_computation_start} to~\ref{line:reach_computation_end}) for each block in \textit{blocks}.
Line~\ref{line:reach_inputs} computes the current input convolution, which is added in line~\ref{line:reach_states_plus_inputs}.
The function \approximate\ overapproximates its argument (a two-dimensional lazy set) to a polygon in constraint representation using Lotov's method (see \sect{sec:set_approx}).

\smallskip

Since vectors in the plane can be ordered by the angle with respect to the positive real axis, we can efficiently evaluate the support vector of a polygon in constraint representation by comparing normal directions, provided that its edges are ordered. We use the symbol $\preceq$ to compare directions, where the increasing direction is counter-clockwise.
The following lemma provides an algorithm to find the support vector.

\begin{lemma}\label{lemma:fastLP}	
Let $\X$ be a polygon described by $m$ linear
constraints $a_i^\transp x \leq b_i$, ordered by the normal vectors $(a_i)$, i.e., $a_{i}\preceq a_{i+1}$ for all $i\in \{1,\ldots,m\}$, where we identify $a_{m+1}$ with $a_1$. Let $\ell \in \R^2 \setminus \{\norigin{2}\}$.
Then there exists $i \in \{1,\dots,m\}$ such that $a_i \preceq \ell \preceq a_{i+1}$ and every optimal solution
$\bar{x}$ of the linear program $\rho_\X(\ell) = \max\{ \ell^\transp x : x \in \X\}$ satisfies $\bar{x} \in \{x : a_i^\transp x \leq b_i\} \cap \{x : a_{i+1}^\transp x \leq b_{i+1}\}.$
\end{lemma}

For the evaluation (\sect{sec:benchmarks}) we use a box approximation. We note that our implementation of the \approximate\ function works for general set approximations; for box approximation we could also use an optimized implementation, e.g., using interval arithmetic. The fact that we approximate in line~\ref{line:reach_inputs} is a design decision; in principle we could keep the elements of $\set{\hat{V}_{\text{tmp}}}$ a lazy set, but experiments have shown that the precision gain is marginal.

\smallskip

We note again that we use a different implementation for models with dense and sparse matrices $\Phi$, respectively.
For instance, the loop around line~\ref{line:reach_innermost} only has to be executed if the submatrix $\Phi^k[b_i, b_j]$ is non-zero.

\subsection{Projection onto Output Variables}

After \reach~has terminated, returning an array of Cartesian products of two-dimensional sets, we usually need to observe some output variables $y(t)$, as in the LTI system Eq.~\eqref{eq:output_system}.
A list with the output variables, or more generally, a \mbyn{2}{n} projection matrix, can be passed to the function \projectandplot. This projection matrix is used if we want to observe, e.g., a linear combination of states. However, for checking safety properties (see below) we do the projection on-line to terminate early if we detect a possible violation.

\subsection{Safety Property Checking}
\label{sec:safety_property_checking}

For checking safety properties, we can improve Algorithm~\ref{algorithm:reach}.
Consider a six-dimensional model with the property $2 x_1 - 3 x_5 < 10$.
A naive approach would compute the reachable states for blocks~1 and~3, i.e., upper and lower bounds for~$x_1$, $x_2$, $x_5$, and~$x_6$.
However, we are only interested in the upper bound for~$x_1$ and the lower bound for~$x_5$.
We modify the algorithm in two ways: First, we replace line~\ref{line:reach_result} by a function that computes the support for the direction of interest. Second, we make the \approximate\ function in line~\ref{line:reach_states_plus_inputs} the identity (i.e., keep the lazy set).
	The reason is that we can evaluate the support directly on the lazy set, so there is no need for an additional overapproximation.

\section{Evaluation}
\label{sec:benchmarks}
%
\begin{table*}
	\begin{threeparttable}
	\renewcommand{\arraystretch}{1}
	\caption{%
	Reach tube computation in dense time.
	The number of time steps is $\num{2e4}$ with step size $\delta = \num{1e-3}$ for both  \tool and \spaceex.
	}
	\label{table:reachability_benchmark}
	\begin{tabular}{l  r  c  c  c  c  P{2}  c  c  P{2}  P{3}}
		\toprule
		\multirow{2}{*}{Model} & \multicolumn{1}{c}{\multirow{2}{*}{$n$}} & \multirow{2}{*}{Var.} & Discretize & \multicolumn{3}{c}{Runtime (sec) one state variable} & \multicolumn{3}{c}{Runtime (sec) all state variables} & \multicolumn{1}{c}{\multirow{2}{*}{O.A.\ \%}} \\
		\cmidrule(lr){5-7} \cmidrule(lr){8-10}
		& & & (sec) & \multicolumn{1}{c}{\tool} & \multicolumn{1}{c}{\spaceex} & \multicolumn{1}{c}{Acc.} & \multicolumn{1}{c}{\tool} & \spaceex & \multicolumn{1}{c}{Acc.} & \\
		\midrule
		Motor & 8 & $x_5$ & $\num{4.89e-4}$ & $\num{1.06}$& $\num{1.90}$ & 1.8 & $\num{4.46}$ & $\num{9.29}$ & 2.1 & 21.53 \\
		Building & 48 & $x_{25}$ & $\num{9.20e-3}$ & $\num{4.49}$ & $\num{9.54}$ & 2.1 & $\num{1.15e2}$ & $\num{2.24e2}$ & 1.9 & 6.50 \\
		PDE & 84 & $x_1$ & $\num{3.30e-2}$ & $\num{4.43}$ & $\num{6.17e1}$ & 13.9 & $\num{1.58e2}$ & $\num{4.75e3}$ & 30.1 & 81.59 \\
		Heat & 200 & $x_{133}$ & $\num{2.09e-1}$ & $\num{2.47e1}$ & $\num{1.02e2}$ & 4.1 & $\num{2.32e3}$ & $\num{5.68e3}$ & 2.4 & 0.05 \\
		ISS & 270 & $x_{182}$ & $\num{2.03e-1}$ & $\num{2.46}$ & $\num{7.91e1}$ & 32.1 & $\num{1.60e2}$ & $\num{8.12e3}$ & 50.8 & 14.52 \\
		Beam & 384 & $x_{89}$ & $\num{1.28}$ & $\num{5.40e1}$ & $\num{3.32e2}$ & 6.1 & $\num{6.81e3}$ & $\num{3.80e4}$ & 5.6 & -30.35 \\
		MNA1 & 578 & $x_1$ & $\num{6.16}$ & $\num{1.40e2}$ & \multicolumn{1}{c}{\crash} & \multicolumn{1}{c}{n/a} & $\num{1.80e4}$ & \multicolumn{1}{c}{\crash} & \multicolumn{1}{c}{n/a} & \multicolumn{1}{c}{n/a} \\
		FOM & 1006 & $x_1$ & $\num{4.70}$ & $\num{1.06e1}$ & \multicolumn{1}{c}{\crash} & \multicolumn{1}{c}{n/a} & $\num{2.92e3}$ & \multicolumn{1}{c}{\crash} & \multicolumn{1}{c}{n/a} & \multicolumn{1}{c}{n/a} \\
		MNA5 & 10913 & $x_1$ & $\num{3.68e2}$ & $\num{1.38e3}$ & \multicolumn{1}{c
		}{\crash} & \multicolumn{1}{c}{n/a} & \multicolumn{1}{c}{\timeout} & \multicolumn{1}{c}{\crash} & \multicolumn{1}{c}{n/a} & \multicolumn{1}{c}{n/a} \\
		\bottomrule
	\end{tabular}
	\begin{tablenotes}
        \footnotesize
        	``Discretize'' stands for the discretization time.
 ``Runtime'' stands for the total runtime.
	 ``Acc.'' stands for acceleration.
	 ``O.A.\ \%'' stands for overapproximation in percent, which is computed as the increase in the bounds computed with \tool for the variable reported under ``Var.'', measured at the last time step, relative to the \spaceex bounds.
	 ``\crash'' marks a crash and 
	 ``\timeout'' marks a timeout ($\num{e5}$ sec).
    \end{tablenotes}
    \end{threeparttable}
	\renewcommand{\arraystretch}{1}
\end{table*}

We evaluate our implementation from \sect{sec:implementation} called \tool on a set of SLICOT benchmark models~\cite{benner1999slicot, chahlaoui2005benchmark, TranNJ16}. They reflect ``real world'' applications with dimensions ranging from eight to over 10,000.
Although some of the original models are differential algebraic equations (DAEs), we have only kept the ODE part, i.e., the coefficient matrices $A$ and $B$, which is consistent with related literature on reach set approximation.
We have performed the evaluation on a notebook with an Intel i5 3.50~GHz CPU and 16~GB RAM running Linux, and we used Julia~v0.6.

\subsection{Reach Tube Benchmarks} \label{ssec:reachability_benchmarks}
We compare \tool to the state-of-the-art support function algorithm LGG implemented in \spaceex.
This algorithm works with template polyhedra, and allows to define the directions that are evaluated. We have considered two cases: one dimension, where we only compute the reach tube in one variable, or full dimensions, where the whole reach tube is computed. 
Note that for implementation reasons, we actually compute the reach set for at least one block (two variables). In the 1D comparison, this means that we compute more information than necessary, while \spaceex truly computes the bounds for a single variable only. In that sense, the comparison is biased in favor of \spaceex.
The reachability results are given in Table~\ref{table:reachability_benchmark}.
\ifextendedversion{The reach tube plots are shown in Appendix~\ref{sec:appendix_reach_plots}.}\fi

To compare the precision, we chose the last time step and compared the bounds for the single variable reported in the table, for \tool and \spaceex, where the \spaceex bounds are the baseline.
For most models the precision is moderately below that of \spaceex. For the PDE model, the approximation error is quite high. For the beam model our analysis is not only faster but also more precise.
In general, we would expect a lower precision than \spaceex for two reasons.
1)~Our reach tube is a Cartesian product of 2D sets; this induces an error that is inherent to the decomposition method, as explained in \sect{ssec:approximation_error}.
2)~\spaceex uses a forward-backward interpolation model, which is more sophisticated than the forward-only model from \sect{sec:approx_model}; we note that our method could also use the \spaceex model without requiring any other changes.

For all models tested, we observe a speedup; as expected, the improvement is more evident for large and sparse models.
For the largest models, \spaceex crashed with a segmentation fault or terminated with a warning that the model dimension is too high.

\subsection{Safety Property Benchmarks} \label{ssec:safety_property_benchmarks}

\begin{table*}
	\begin{threeparttable}
	\renewcommand{\arraystretch}{1}
	\caption{%
	Verification of safety properties in discrete time.
	The number of time steps is $\num{4e3}$ with step size $\delta = \num{5e-3}$ for both \tool and \hylaa.
	}
	\label{table:safety_property_benchmark}
	\begin{tabular}{l r >{\centering}m{27mm} c c c c P{2} c}
		\toprule
		\multirow{3}{*}{Model} & \multicolumn{1}{c}{\multirow{3}{*}{$n$}} & \multirow{3}{*}{Property} & \multicolumn{5}{c}{Runtime (sec)} & \multirow{3}{*}{Verified} \\
		\cline{4-8}
		& & & \multicolumn{3}{c}{\tool} & \multicolumn{1}{c}{\multirow{2}{*}{\hylaa}} & \multicolumn{1}{c}{\multirow{2}{*}{Acc.}} & \\
		\cline{4-6}
		& & & Discretize & Check & Total & & & \\
		\midrule
		Motor & 8 & \mbox{$x_1 \notin [0.35, 0.4]$} $\lor \, x_5 \notin [0.45, 0.6]$ & $\num{4.5e-4}$ & $\num{2.48e-1}$ & $\num{2.48e-1}$ & $\num{1.6}$ & 6.5 & \yes \\
		Building & 48 & $x_{25} < \num{6e-3}$ & $\num{9.87e-3}$ & $\num{5.20e-1}$ & $\num{5.30e-1}$ & $\num{2.5}$ & 4.7 & \yes \\
		PDE & 84 & $y_1 < 12$ & $\num{1.62e-2}$ & $\num{2.22e1}$ & $\num{2.22e1}$ & $\num{3.5}$ & 0.2 & \yes \\
		Heat & 200 & $x_{133} < 0.1$ & $\num{1.48e-1}$ & $\num{4.08}$ & $\num{4.23}$ & $\num{1.38e1}$ & 3.3 & \yes \\
		ISS & 270 & $y_3 \in [-7,7] \times 10^{-4}$ & $\num{1.87e-1}$ & $\num{2.12e1}$ & $\num{2.14e1}$ & $\num{1.53e2}$ & 7.1 & \no \\
		Beam & 384 & $x_{89} < 2100$ & $\num{3.66e-1}$ & $\num{6.60}$ & $\num{6.97}$ & $\num{1.69e2}$ & 24.2 & \yes \\
		MNA1 & 578 & $x_1 < 0.5$ & $\num{1.54}$ & $\num{1.82e1}$ & $\num{1.97e1}$ & $\num{2.88e2}$ & 14.6 & \yes \\
		FOM & 1006 & $y_1 < 185$ & $\num{4.48}$ & $\num{4.56e2}$ & $\num{4.60e2}$ & $\num{3.30e2}$ & 0.7 & \no \\
		MNA5 & 10913 & \mbox{$x_1 < 0.2$} $\land \, x_2 < 0.15 $ & $\num{2.32e2}$ & $\num{2.03e2}$ & $\num{4.35e2}$ & $\num{3.44e4}$ & 79.1 & \yes \\
		\bottomrule
	\end{tabular}
	\begin{tablenotes}
		\footnotesize
		The $y$ variables denote output variables, as in~\eqref{eq:output_system}, consisting of linear combinations of state variables (involving all variables for PDE/FOM and half of the variables for ISS).
		``Acc.'' stands for acceleration.
		The last column shows if we could verify the property for the given time step.
	\end{tablenotes}
	\end{threeparttable}
	\renewcommand{\arraystretch}{1}
\end{table*}

As described in \sect{sec:safety_property_checking}, we can check safety properties in the form of (conjunctions and disjunctions of) linear inequalities over the state variables.
In Table~\ref{table:safety_property_benchmark} we compare our results to those of \hylaa~\cite{bak2017simulation}, a simulation-based verification tool in discrete time.

\hylaa assumes that the inputs are constant between time steps, and we stick to this assumption for the purpose of comparison. We used the same time step as in the evaluation of~\cite{bak2017simulation} and were able to verify all safety properties except for the models ISS and FOM.
With a bigger time step, we can also verify those properties.
\hylaa verified all benchmarks.
We had to modify the \hylaa code (reduced the time horizon chunk size \texttt{max\_steps\_in\_mem} from~527 to~400) for the FOM model to prevent out-of-memory problems.

\begin{table}[t]
	\begin{threeparttable}
	\renewcommand{\arraystretch}{1}
	\caption{%
	Verification of safety properties in dense time. 
	}
	\label{table:safety_property_benchmark_dense}
	\begin{tabular}{@{\hspace*{1mm}} l @{\hspace*{1mm}} r >{\centering}m{23mm} c c @{\hspace*{1mm}}}
		\toprule
		Model & \multicolumn{1}{c}{$n$} & Property & $\delta$ & Runtime (sec) \\
		\midrule
		Motor & 8 & \mbox{$x_1 \notin [0.35, 0.4]$} $\lor \, x_5 \notin [0.45, 0.6]$ & $\num{1e-3}$ & $\num{1.62}$ \\
		Building & 48 & $x_{25} < \num{6e-3}$ & $\num{3e-3}$ & $\num{8.76e-1}$ \\
		PDE & 84 & $y_1<12$ & $\num{3e-4}$ & $\num{1.03e3}$ \\
		Heat & 200 & $x_{133} < 0.1$ & $\num{1e-3}$ & $\num{1.48e1}$ \\
		Beam & 384 & $x_{89} < 2100$ & $\num{5e-5}$ & $\num{8.57e2}$ \\
		MNA1 & 578 & $x_1 < 0.5$ & $\num{4e-4}$ & $\num{2.87e2}$ \\
		MNA5 & 10913 & $x_1 \!\!<\!\! 0.2 \land x_2 \!\!<\!\! 0.15$ & $\num{3e-1}$ & $\num{3.92e2}$ \\
		\bottomrule
	\end{tabular}
\begin{tablenotes}
	\footnotesize
	Step sizes are selected such that the property is satisfied.
	The time horizon is 20.
\end{tablenotes}
	\end{threeparttable}
	\renewcommand{\arraystretch}{1}
\end{table}

\smallskip

We also applied \tool to the benchmarks in dense time and were also able to verify all properties except for ISS and FOM.
In Table~\ref{table:safety_property_benchmark_dense} we show the results.
In particular, we were able to verify the property of the MNA5 model with 10,000 variables in less than 7~minutes, where 98\% of the time was spent in the discretization.

\subsection{Discussion}

Each LTI system has its own structural properties, describing how each state influences the dynamics of the system.
The SLICOT models have different sparsity patterns%
\ifextendedversion{ (see Appendix~\ref{sec:appendix_sparsity} for details)}\fi%
, which we have exploited effectively in Algorithm~\ref{algorithm:reach}. In our context it is natural to measure the sparsity of $\Phi$ as the number of \nbyn[2] blocks with at least one non-zero element, divided by the total number of blocks ($b^2$). As a rule of thumb, for a given row-block the cost increases linearly in the number of occupied blocks. For models such as Heat and Beam, the sparsity is 0\%, meaning that the matrix is completely dense, while for models such as ISS and FOM it is 97.8\% and 99.8\%, respectively. We note that the matrix power operation does not necessarily preserve the sparsity pattern, although it does in some particular cases, e.g., if $\Phi$ is block upper-triangular.

The efficiency with respect to the sparsity pattern is manifest in the small runtimes for sparse models, compared to higher runtimes for dense models. In contrast, non-decomposed methods cannot make full use of the sparsity since they rely on a high-dimensional LP even for evaluating the support vector in a single direction. This explains the very high speedup of $\times 50$ for ISS. Moreover, the $\num{1006}$-dimensional FOM model is analyzed in about the same time.

For the discrete evaluation, with the same step size~$\delta$ as in \hylaa, for seven out of the nine examples we observe a speedup which ranges from $\times 3$ up to $\times 79$. A crucial difference of this scenario with respect to dense time reachability is that the property is satisfied for larger~$\delta$.
As expected, our approach scales best for the models whose properties only involve a few variables; PDE/FOM involve all variables, and here \hylaa is faster;
ISS involves half of the variables, and here we still achieve a speedup of factor~7.

Let us remark that in the cases where \tool is not precise enough, namely ISS and FOM, the property involves an output. The effect of a higher error for linear combinations than for single blocks is reasonable, since in these experiments we have only considered box directions. An alternative approach for this use case, which we have not investigated yet, would be to use a refined 2D \approximate~function that introduces more constraints to the polygonal approximation of the reach set.

\section{Conclusions}
\label{sec:conclusions}
%
We have revisited the fundamental set-based recurrence relation that arises in the study of reachability problems with affine dynamics and nondeterministic inputs.
For this we combined high dimensional linear algebra with low dimensional set computations and a state-of-the-art reachability algorithm.
We have shown that this approach is advantageous against the ``curse of dimensionality'': Reformulating the recurrence as a sequence of independent low-dimensional problems, where the set-based computations can be performed efficiently, we can effectively scale to high order systems.
The overapproximation is conservative due the decomposition, and we have characterized the influence of initial states, inputs, dynamics, and time step with an analytical upper bound.

We have evaluated our method on a set of real-world models from control engineering, involving many coupled variables. Numerical results show a speedup of up to two orders of magnitude with respect to state-of-the-art approaches that are non-decomposed. Apart from one exception, the overapproximation is within 22\% of the non-decomposed solution. 
In the dense-time case, our approach can handle systems with substantially more variables than the state-of-the-art tool \spaceex, by almost two orders of magnitude.
Note that in this paper we have only tested box directions to represent two-dimensional sets, since the accuracy seemed sufficient. The investigation of a method producing more accurate low-dimensional projections, arbitrarily close to the exact projection, could deliver even more precise results.

The approach presented in this paper can benefit from parallelization: the computations for each block are completely independent (see the loop in line~\ref{line:reach_computation_start} of Algorithm~\ref{algorithm:reach}). Using a separate thread for each block, this will give a speedup of~$n/2$.
\spaceex can also be parallelized, for bounding boxes by a factor of~$2n$. Comparing a parallelized Algorithm~\ref{algorithm:reach}) with parallelized \spaceex, we could theoretically see the speed-up in Table~\ref{table:reachability_benchmark} reduced from $50\times$ to $12\times$ (ISS benchmark).
However, in practice, the speed-up from parallelizing the LGG algorithm used in \spaceex turns out much more modest \cite{ray2015xspeed}.
We have only discussed partitions of two-dimensional blocks, and sequentially for each pair of rows. However, there is no theoretical restriction in considering blocks of dimensions one or three for the explicit computations, which should give further gains in speed for the former and gains in precision for the latter.
Furthermore, allowing overlapping blocks leads to relative completeness for software~\cite{HoenickeMP17}.

Similarity transformations, such as Schur or Jordan transformations, could be applied to the system's dynamics just after the discretization.
This will eventually have an impact on the accumulated error, on the performance (since the number of non-zero blocks would change), or both. Characterizing the advantage of using similarity transformations for a given dynamics matrix, initial states, and inputs is left for future study.

\begin{acks}
	M.F.\ acknowledges stimulating discussions with Alexandre Rocca and Cesare Molinari.
	
	This work was partially supported by the
	\grantsponsor{1}{European Commission}{https://cps-vo.org/group/UnCoVerCPS}
	under grant no.\
	\grantnum{1}{643921}
	(UnCoVerCPS), by the
	\grantsponsor{2}{Metro Grenoble}{}
	through the project NANO2017,
	by the
	\grantsponsor{3}{Air Force Office of Scientific Research}{http://www.wpafb.af.mil/afrl/afosr/}
	under award no.\
	\grantnum{3}{FA2386-17-1-4065},
	and by the ARC project
	\grantnum{4}{DP140104219}
	(\grantsponsor{4}{Robust AI Planning for Hybrid Systems}{https://cs.anu.edu.au/research/research-projects/robust-ai-planning-hybrid-systems}).
	Any options, finding, and conclusions or
	recommendations expressed in this material are those of the authors
	and do not necessarily reflect the views of the United States Air
	Force.
\end{acks}

\newpage
\bibliographystyle{ACM-Reference-Format}
\bibliography{bibliography.bib}

\ifextendedversion{%
\clearpage
\appendix
%
\section{Proofs}
\subsection{Proposition \ref{prop:HDistDecompMap}}
\begin{proof}[Proof of Proposition \ref{prop:HDistDecompMap}]
	The support function of $\set{\bar X}'$ on $d \in \Reals^n$ is, applying the properties in Lemma~\ref{lemma:support_function_and_vector_properties},
\begin{align*}
	\rho_{\set{\bar X}'}(d) &= \rho_{\Phi \set{\hat X}} (d) = \rho_{\set{\hat X}_1\times \cdots \times \set{\hat X}_b} (\Phi^\transp d) \\
	&= \sum_j \rho_{ \set{\hat X}_j}\left(\pi_j (\Phi^T d)\right) 
=  \sum_j \rho_{ \set{\hat X}_j}\left(\sum_k \Phi^\transp_{kj} d_k\right).
\end{align*}
On the other hand,
\begin{align*}
\rho_{\Xhat'}(d) &= \rho_{\Xhat'_1\times \cdots \times \Xhat'_b} (d) 
= \sum_i \rho_{\Xhat'_i}\left(d_i\right) =  \sum_{i, j} \rho_{\Xhat_j}\left(\Phi^T_{ij} d_i\right).
\end{align*}
The result is obtained plugging these expressions into~\eqref{eq:dpH_suppFun}.
\end{proof}
\subsection{Proposition \ref{prop:HDistDecompMap_simplified3}}
For the proof of Proposition \ref{prop:HDistDecompMap_simplified3} we need the following intermediate result.

In the following formula we can reduce the bound on the approximation error by freely selecting one specific row-block for each column-block of $\Phi$. In the $j$-th column-block, we denote this selection by $q_j$.
\begin{lemma}\label{prop:HDistDecompMap_simplified}
The approximation error is bounded by
\begin{equation*}\label{eq:HDistDecompMap_simplified}
d_H^p(\set{\bar X}', \set{\hat{X}'}) \leq \max_{\Vert d \Vert_p \leq 1} \sum_j \sum_{i \neq q_j} \rho_{ \Xhat_j}(\Phi_{ij}^T d_i) + \rho_{ \Xhat_j}\left(- \Phi_{ij}^T d_i \right).
\end{equation*}
\end{lemma}

\begin{proof}
	We use the property of support functions that $$\rho_\set{X}(u+v) \geq \rho_\set{X}(u) - \rho_\set{X}(-v).$$
	With this we can bound, picking any $q \in 1,\ldots,b$,
	\begin{align*}
	\rho_{ \Xhat_j}\left(\sum_k \Phi_{kj}^T d_k \right) \geq 
	\rho_{ \Xhat_j}\left(\Phi_{qj}^T d_q \right) - \sum_{k \neq q}   \rho_{ \Xhat_j}\left(- \Phi_{kj}^T d_k \right).
	\end{align*}
	We let~$q$ be a function of $j$ and substitute the above in~\eqref{eq:error_crossprod} with $k:=i$:
	\begin{align*}
	d_H^p(\set{\bar X}', \Xhat') &\leq  \max_{\Vert d \Vert_p \leq 1} \sum_{j}  \sum_{i} \rho_{ \set{\hat X}_j}(\Phi_{ij}^T d_i) -
	\rho_{ \set{\hat  X}_j}\left(\Phi_{q_jj}^T d_{q_j} \right)\\
	&\quad + \sum_{i \neq q_j}   \rho_{ \set{\hat  X}_j}\left(- \Phi_{ij}^T d_i \right)\\
	&=  \max_{\Vert d \Vert_p \leq 1} \sum_{j}  \sum_{i \neq q_j} \rho_{ \set{\hat  X}_j}(\Phi_{ij}^T d_i) + \rho_{ \set{\hat  X}_j}\left(- \Phi_{ij}^T d_i \right)\\
	&\quad - \rho_{ \set{\hat  X}_j}\left(\Phi_{q_jj}^T d_{q_j} \right) + \rho_{ \set{\hat  X}_j}\left(\Phi_{q_jj}^T d_{q_j} \right)\\
	&= \max_{\Vert d \Vert_p \leq 1} \sum_{j}  \sum_{i \neq q_j} \rho_{ \set{\hat  X}_j}(\Phi_{ij}^T d_i) + \rho_{ \set{\hat  X}_j}\left(- \Phi_{ij}^T d_i \right).
	\end{align*}
\end{proof}
The bound is actually tight. 
Intuitively speaking, the judicious selection of $q_j$ allows us to eliminate the block $\Phi_{ij}$ that contributes most to the error bound. 
Then we can bound the approximation error by
\begin{equation}\label{eq:HDistDecompMap_simplified2}
\hausdorff^p\bigl(\set{\bar X}',\set{\hat X}'\bigr) \leq 
\max_{\begin{tabular}{c}$\scriptstyle d\,=\,d_1 \times \ldots \times d_b$ \\[-2pt] $\scriptstyle \norm{d}_p\,\leq\, 1$\end{tabular}}
\sum_{j} \sum_{i\neq q_j} \norm{\Phi_{ij}^Td_i}_{\tfrac{p}{p-1}} \Delta_j.
\end{equation}

\begin{proof}[Proof of Proposition \ref{prop:HDistDecompMap_simplified3}]
	First, we apply to~\eqref{eq:HDistDecompMap_simplified2} that $$\norm{\Phi_{ij}^Td_i}_{\tfrac{p}{p-1}} \leq \norm{\Phi_{ij}^T}_{\tfrac{p}{p-1}} \norm{d_i}_{\tfrac{p}{p-1}} = \norm{\Phi_{ij}}_{p} \norm{d_i}_{\tfrac{p}{p-1}}.$$ This gives us
	\begin{align*}\label{eq:normstuff}
	\hausdorff^p\bigl(\set{\bar X}',\set{\hat X}'\bigr) &\leq 
	\max_{\begin{tabular}{c}$\scriptstyle d\,=\,d_1 \times \ldots \times d_b$ \\[-2pt] $\scriptstyle \norm{d}_p\,\leq\, 1$\end{tabular}}
	\sum_{j} \sum_{i\neq q_j} \alpha_j \norm{d_i}_{\tfrac{p}{p-1}} \Delta_j \\
	&= \max_{\begin{tabular}{c}$\scriptstyle d\,=\,d_1 \times \ldots \times d_b$ \\[-2pt] $\scriptstyle \norm{d}_p\,\leq\, 1$\end{tabular}}
	\sum_{j} \alpha_j  \Delta_j  \sum_{i\neq q_j}  \norm{d_i}_{\tfrac{p}{p-1}}
	\end{align*}
	With $\sum_{i\neq q_j} \norm{d_i}_{\tfrac{p}{p-1}} \leq (b-1)\norm{d}_{\tfrac{p}{p-1}}$ we get
	\begin{equation*}\label{eq:HDistDecompMap_simplified3_pp}
	\hausdorff^p\bigl(\set{\bar X}',\set{\hat X}'\bigr) \leq 
	\max_{\begin{tabular}{c}$\scriptstyle d\,=\,d_1 \times \ldots \times d_b$ \\[-2pt] $\scriptstyle \norm{d}_p\,\leq\, 1$\end{tabular}}
	(b-1) \norm{d}_{\tfrac{p}{p-1}} \sum_{j} \alpha_j  \Delta_j  .
	\end{equation*}
	For $p \leq 2$, it is known that $\norm{d}_{\tfrac{p}{p-1}} \leq  \norm{d}_p,$ which leads  to~\eqref{eq:HDistDecompMap_simplified3_ppetit}.
	The result holds for any $p \geq 1$ since $\norm{x}_1 \geq \norm{x}_p$.
\end{proof}

\subsection{Proposition \ref{prop:aff_rec_error}}
\begin{proof}[Proof of Proposition\;\ref{prop:aff_rec_error}]
By Eq.~\eqref{eq:constant_nonrecurrent_decomposed}, and since $\decompose$ distributes over Minkowski sum (Eq.~\eqref{eq:decomp_minkowski}), we get
\[
\Xhat(k) = \decompose(\Phi^k\X(0)) \oplus \What(k).
\]
Then,
\begin{align*}
\hausdorff^p\bigl(\set{\Xhat}(k),\set{X}(k)\bigr) 
= 
\hausdorff^p\bigl(\decompose(\Phi^k\X(0))\oplus \What(k),\Phi^k\X(0)\oplus \W(k)\bigr)
\\
\leq \hausdorff^p\bigl(\decompose(\Phi^k\X(0)), \Phi^k\X(0)\bigr) + \hausdorff^p\bigl(\What(k),\W(k)\bigr).
\end{align*}
Applying Prop.\;\ref{prop:total_aff_err_bound} with $\alpha_j \leq K_\Phi \alpha_\Phi^k$, we get the bound
\begin{align*}
\hausdorff^p\bigl(\decompose(\Phi^k\X(0))&, \Phi^k\X(0)\bigr) \\
&\leq (b-1) \sum_j \alpha_j \Delta_j^x + \Vert \Phi^k \Vert_p \hausdorff^p\bigl(\Xhat(0)), \X(0)\bigr)
\\
&\leq K_\Phi \alpha_\Phi^k (b-1)    \Delta^x_\mathrm{sum} + K_\Phi \alpha_\Phi^k \varepsilon^x.
\end{align*}

Similarly, we get
\begin{align*}
\hausdorff^p\bigl(\set{\What}(k),\set{W}(k)\bigr) &\leq 
\varepsilon^v + K_\Phi ((b-1)\Delta^v_\mathrm{sum} + \varepsilon^v) \sum_{s=1}^{k-1} \alpha_\Phi^s \\
&= \varepsilon^v  + K_\Phi  \bigl( (b-1)    \Delta^v_\mathrm{sum} + \varepsilon^v \bigr) \alpha_\Phi \frac{1-\alpha_\Phi^{k-1}}{1-\alpha_\Phi}.
\end{align*}
The conclusion follows from combining both bounds.
\end{proof}

\section{Support Functions}
\label{sec:additional}

We recall the following elementary properties of support functions and support vectors.

\begin{lemma}\label{lemma:support_function_and_vector_properties}
	For all compact convex sets $\X$, $\Y$ in $\Reals^n$, for all \nbyn real matrices $M$, all scalars $\lambda$, and all vectors $\ell \in \Reals^n$, we have:
	\begin{itemize}
		\item $\rho_{\lambda\X} (\ell) = \rho_{\X} (\lambda \ell)$, $\sigma_{\lambda\X} (\ell) = \lambda \sigma_{\X} (\lambda \ell)$
		 \goodbreak
		 \vspace{0.2cm} 
		\item $\rho_{M\X} (\ell) = \rho_{\X} (M^\transp \ell)$, $\sigma_{M\X} (\ell) = M\sigma_{\X} (M^\transp \ell)$
		 \goodbreak
		 \vspace{0.2cm}
		\item  $\rho_{\X \oplus \Y} (\ell) = \rho_{\X} (\ell) + \rho_{\Y} (\ell)$,  $\sigma_{\X \oplus \Y} (\ell) = \sigma_{\X} (\ell) \oplus \sigma_{\Y} (\ell)$
		 \goodbreak
		 \vspace{0.2cm}
		\item  $\rho_{\X \times \Y} (\ell) = \ell^\transp \sigma_{\X \times \Y}(\ell)$,\\ $\sigma_{\X \times \Y} (\ell) = (\sigma_{\X}(\ell_1), \sigma_{\Y}(\ell_2))$, $\ell = (\ell_1, \ell_2)$
		 \goodbreak
		 \vspace{0.2cm}
		\item  $\rho_{\CH(\X\cup\Y)} (\ell) = \max (\rho_{\X} (\ell), \rho_{\Y} (\ell))$, \\
		$\sigma_{\CH(\X\cup\Y)} (\ell) = \argmax\limits_{x, y} (\ell^\transp x, \ell^\transp y), x \in \sigma_{\X}(\ell), y \in \sigma_{\Y}(\ell)$
	\end{itemize}
\end{lemma}

\section{Sparsity patterns}
\label{sec:appendix_sparsity}

\def\pic#1{\centering \includegraphics[keepaspectratio,height=1.3cm]{fig/sparsity/#1}}
\begin{table}[ht]
	\caption{%
		Sparsity characteristics of the SLICOT benchmarks.
		The sparsity of $A$ ($\Phi$), noted as ``$\spA$'' (``$\spPhi$''), is the relative number of non-zero \nbyn[2] blocks.
	}
	\label{table:slicot_benchmark}
	\begin{tabular}{l r c c m{16mm} @{\hspace*{2mm}} m{16mm} @{} c @{}}
		\toprule
		\multicolumn{1}{c}{Model} & \multicolumn{1}{c}{$n$} & $\spA$ & $\spPhi$ & \multicolumn{1}{c}{\parbox[c]{10mm}{Sparsity \hfill plot ($A$)}} & \multicolumn{1}{c}{\parbox[c]{10mm}{Sparsity \hfill plot ($\Phi$)}} & \\
		\midrule
		Motor & 8 & 50.0\% &50.0\% & \pic{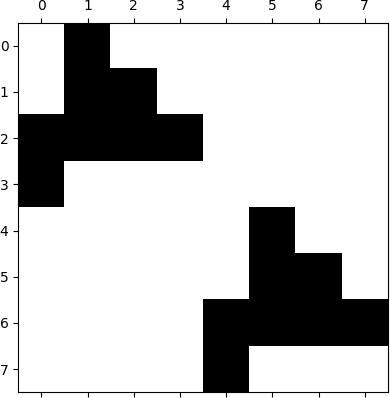} & \pic{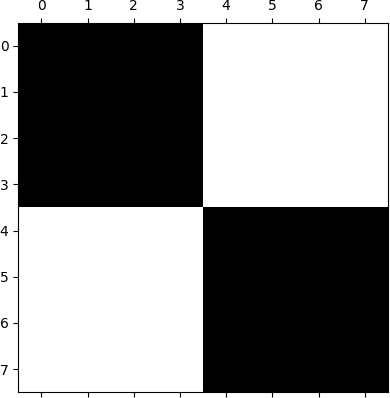} & \\
		Building & 48 & 47.9\% & 0.0\% & \pic{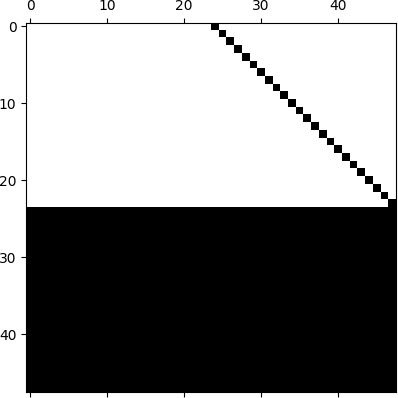} & \pic{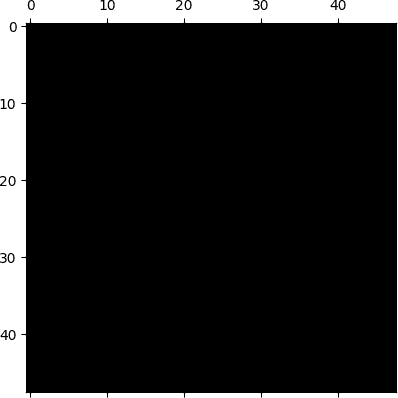} & \\
		PDE & 84 & 84.8\% & 0.0\%  & \pic{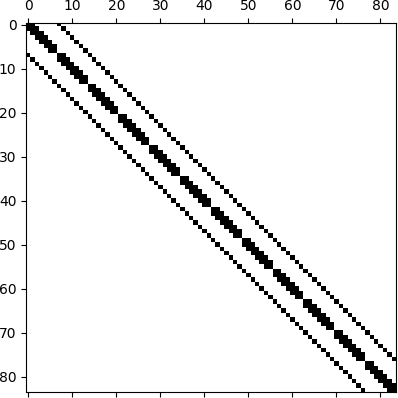} & \pic{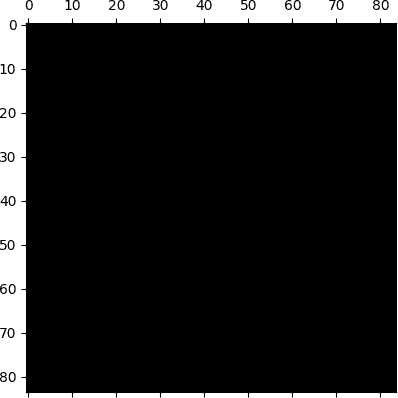} & \\
		Heat & 200 & 97.0\% & 0.0\% & \pic{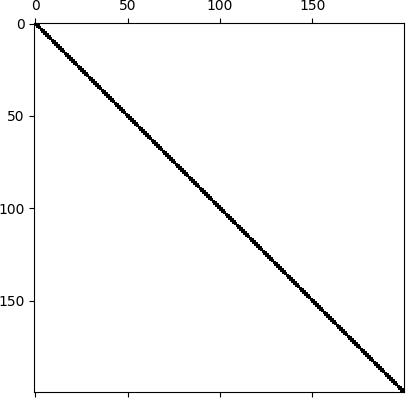} & \pic{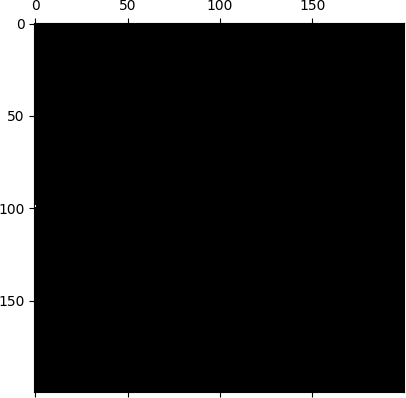} & \\
		ISS & 270 & 98.1\% & 97.7\% & \pic{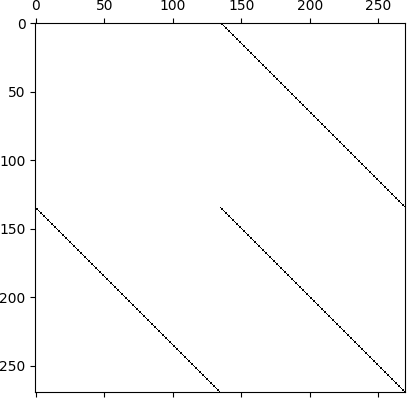} & \pic{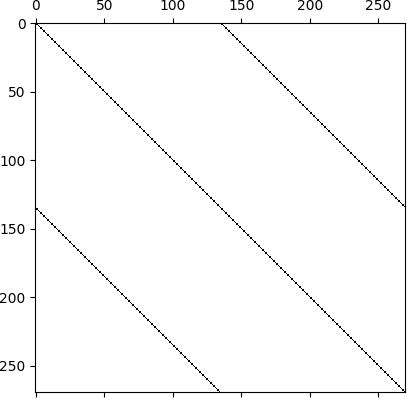} & \\
		Beam & 384 & 49.7\% & 0.0\% & \pic{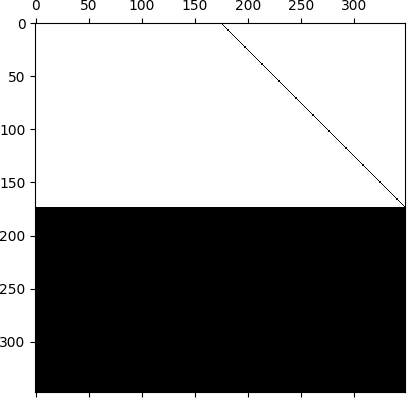} & \pic{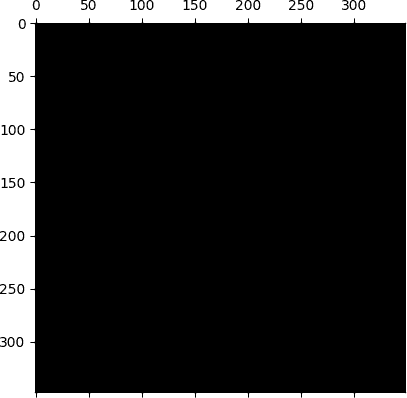} & \\
		MNA1 & 578 & 98.3\% & 3.4\% & \pic{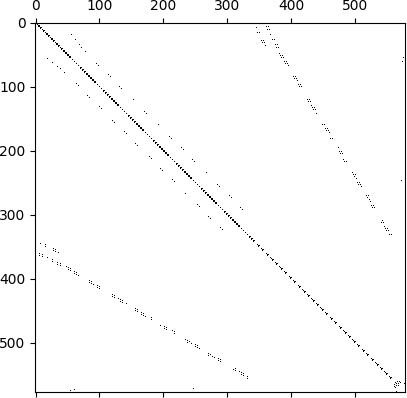} & \pic{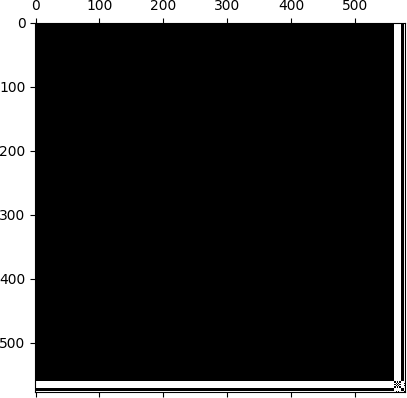} & \\
		FOM & 1006 & 99.8\% & 99.8\% & \pic{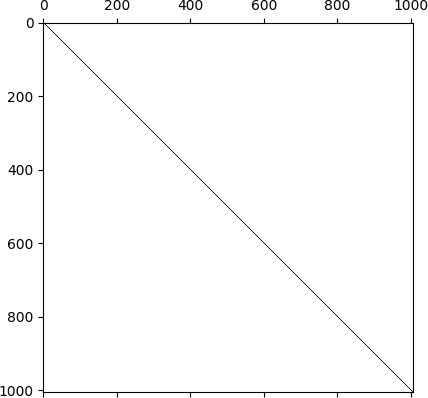} & \pic{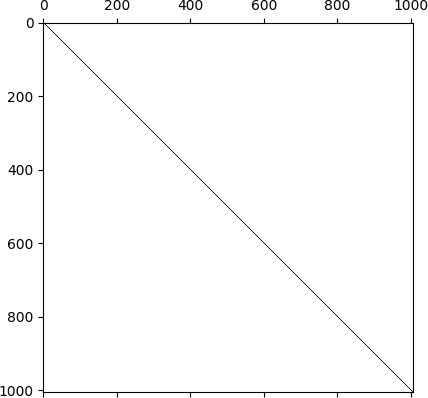} & \\
		MNA5 & 10913 & 99.9\% & 98.8\% & \centering \crash & \multicolumn{1}{c}{\crash} & \\
		\bottomrule
	\end{tabular}
\begin{tablenotes}
	\footnotesize
	For the MNA5 model the plotting engine crashed (\crash).
\end{tablenotes}
\end{table}

\onecolumn
\clearpage
\section{Reach tube plots}
\label{sec:appendix_reach_plots}

\def\d{50mm}
\def\pic#1{\hspace*{0mm} \hfill \includegraphics[keepaspectratio,height=48mm]{#1}}
\begin{longtable}{c m{7cm} m{7cm} m{0mm} @{}}
	\caption{Reach tube plots in dense time for both \tool and \spaceex.
	The x-axis always shows the time and the y-axis shows the variable that we reported in Table~\ref{table:reachability_benchmark}, using the same step size $\delta = \num{1e-3}$.} \\
	\toprule
	\endfirsthead
	\caption{Reach tube plots in dense time (continued).} \\
	\toprule
	\endhead
	\bottomrule
	\endfoot
	\bottomrule
	\endlastfoot
	\multicolumn{1}{c}{Model} & \multicolumn{1}{c}{\tool} & \multicolumn{1}{c}{SpaceEx} \\
	\midrule
	Motor & \pic{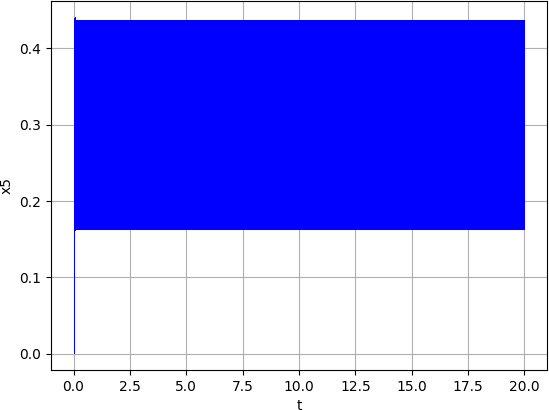} & \pic{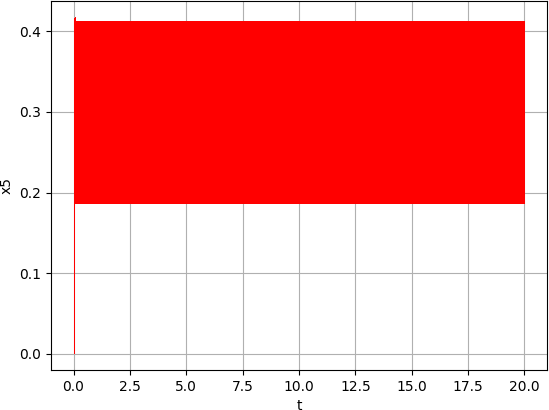} &\\[\d]
	Building & \pic{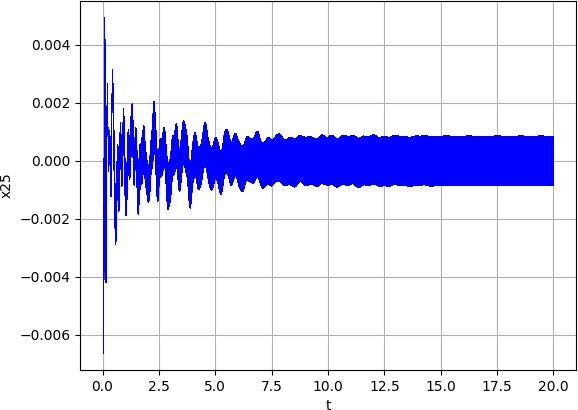} & \pic{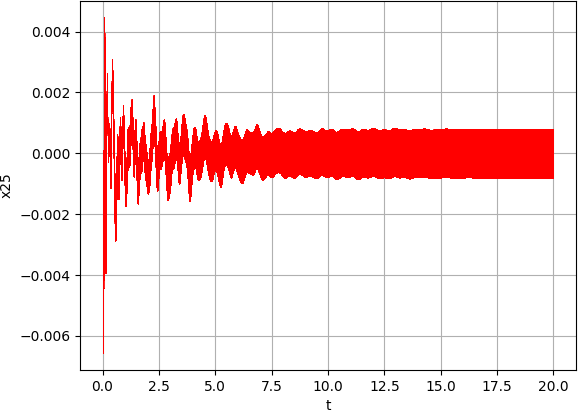} &\\[\d]
	PDE & \pic{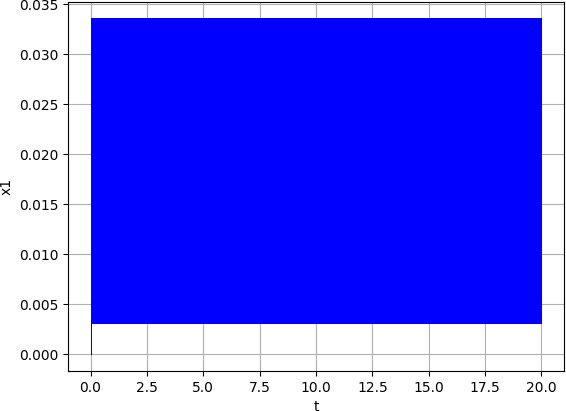} & \pic{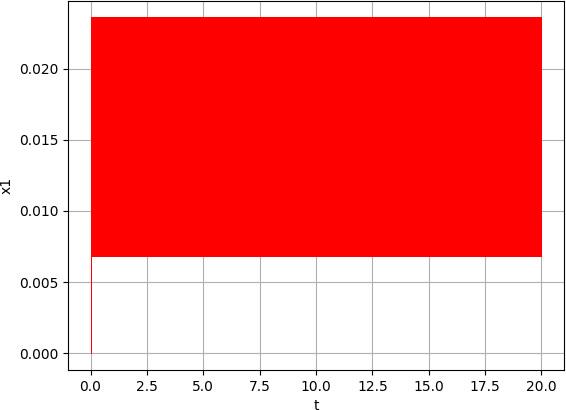} &\\[\d]
	Heat & \pic{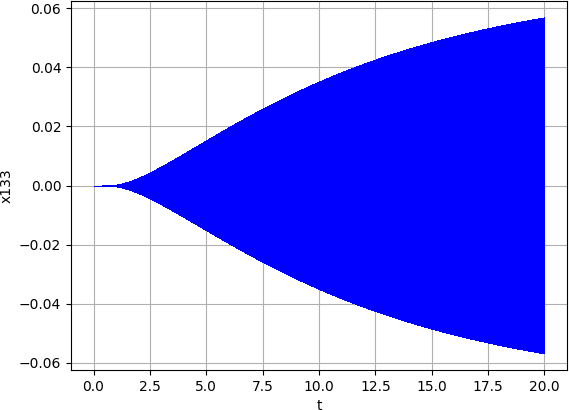} & \pic{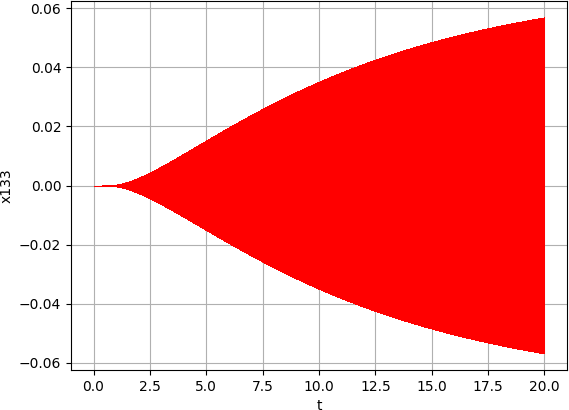} &\\[\d]
	ISS & \pic{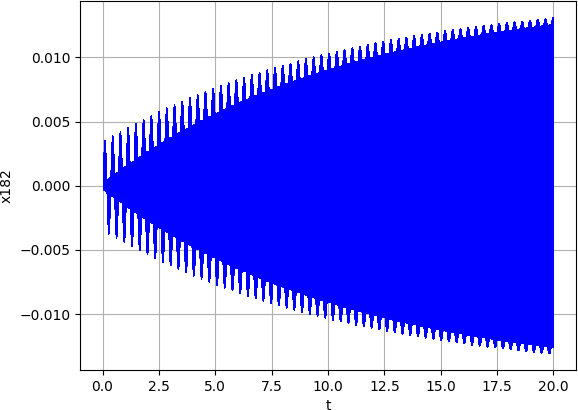} & \pic{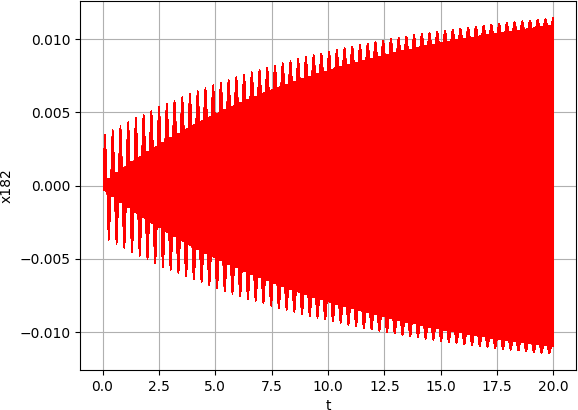} &\\[\d]
	Beam & \pic{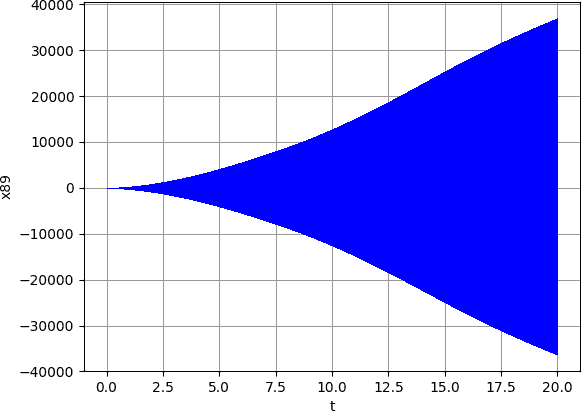} & \pic{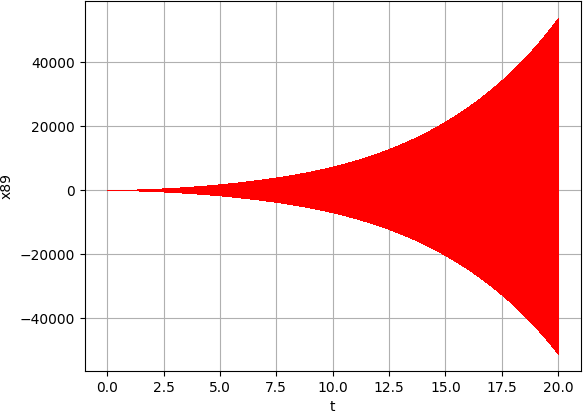} &\\[\d]
	MNA1 & \pic{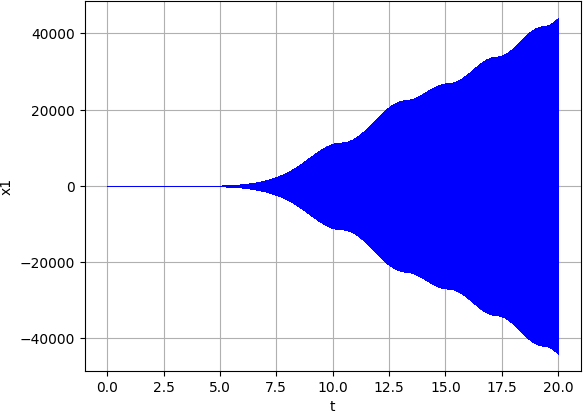} & \multicolumn{1}{c}{n/a} &\\[\d]
	FOM & \pic{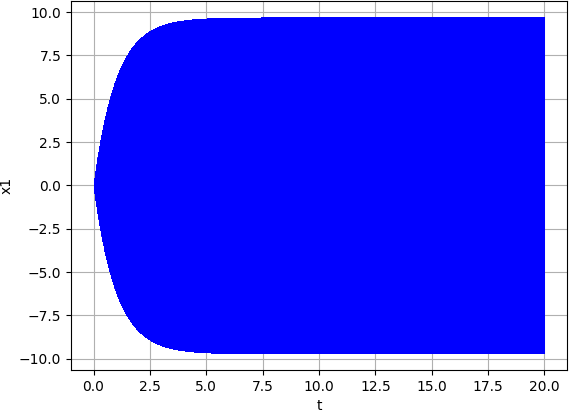} & \multicolumn{1}{c}{n/a} &\\[\d]
	MNA5 & \pic{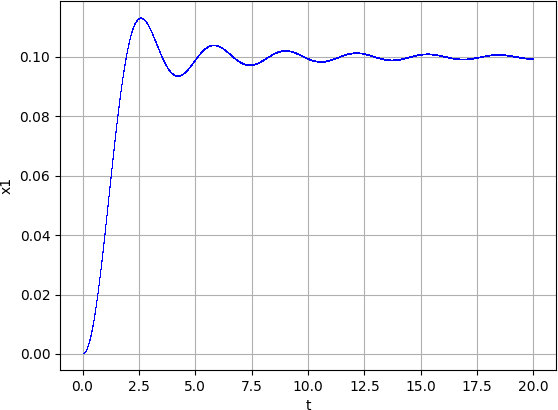} & \multicolumn{1}{c}{n/a} \\
\end{longtable}

}\fi

\end{document}